\documentclass[amsmath,nobibnotes,aps,amssymb,pra, prl,aps,showpacs,superscriptaddress,twocolumn, longbibliography, reprint]{revtex4-2}
\usepackage{amsmath,amsfonts,amssymb,amsthm,graphics,graphicx,epsfig,bbm}
\usepackage[colorlinks=true,citecolor=blue,linkcolor=blue,urlcolor=blue]{hyperref}
\usepackage[usenames]{color}
\usepackage{graphicx}
\usepackage{subfigure}
\usepackage{amsmath}
\usepackage{tikz}
\usetikzlibrary{quantikz}
\usepackage{epsfig}
\usepackage{dcolumn}
\usepackage{bm}
\usepackage{color}
\usepackage{times}
\usepackage{epstopdf}
\usepackage{amssymb}
\usepackage{amstext}
\usepackage{latexsym}
\usepackage{float}
\usepackage{hyperref}
\usepackage{amsfonts}
\usepackage{psfrag}
\usepackage{soul,xcolor}
\usepackage[normalem]{ulem}
\usepackage{dsfont}
\usepackage{txfonts}
\usepackage{physics}
\usepackage{footnote}
\usepackage{multirow}
\usepackage{appendix}
\usepackage{mathtools}
\usepackage{xspace}
\usepackage{algorithm}
\usepackage{algpseudocode}

\theoremstyle{definition}

\newtheorem{theorem}{Theorem}
\newtheorem*{lemma*}{lemma}
\newtheorem*{corollary*}{Corollary}

\usepackage{mathtools}

\begin{document}
\setstcolor{red}
\newtheorem{Proposition}{Proposition}[section]
\title{Deep Thermalization and Measurements of Quantum Resources  }
\date{\today}

\author{Naga Dileep Varikuti}
\email{dileep.varikuti@unitn.it}
\affiliation{Pitaevskii BEC Center, CNR-INO and Dipartimento di Fisica, Universit\`a di Trento, Via Sommarive 14, Trento, I-38123, Italy}
\affiliation{INFN-TIFPA, Trento Institute for Fundamental Physics and Applications, Via Sommarive 14, Trento, I-38123, Italy}
\author{Soumik Bandyopadhyay}
\email{soumik@iisertvm.ac.in}
\affiliation{Pitaevskii BEC Center, CNR-INO and Dipartimento di Fisica, Universit\`a di Trento, Via Sommarive 14, Trento, I-38123, Italy}
\affiliation{INFN-TIFPA, Trento Institute for Fundamental Physics and Applications, Via Sommarive 14, Trento, I-38123, Italy}
\affiliation{School of Physics, Indian Institute of Science Education and Research Thiruvananthapuram, Thiruvananthapuram, Kerala 695551, India}
\author{Philipp Hauke}
\email{philipp.hauke@unitn.it}
\affiliation{Pitaevskii BEC Center, CNR-INO and Dipartimento di Fisica, Universit\`a di Trento, Via Sommarive 14, Trento, I-38123, Italy}
\affiliation{INFN-TIFPA, Trento Institute for Fundamental Physics and Applications, Via Sommarive 14, Trento, I-38123, Italy}

\begin{abstract}
Quantum resource theories provide a unified framework for characterizing useful quantum phenomena subject to physical constraints, but are notoriously hard to assess in experimental systems. In this letter, we introduce a unified protocol for quantifying the resource-generating power of arbitrary quantum evolutions applicable to multiple {resource theories}.
It is based on deep thermalization, which has recently gained attention for its role in the emergence of quantum state designs from partial projective measurements.
Central to our approach is the use of projected ensembles, recently employed to probe {deep thermalization}, together with new twirling identities that allow us to directly infer the {resource-generating powers} of the underlying dynamics. These identities further reveal how resources build up and thermalize in generic quantum circuits. Finally, we show that quantum resources themselves undergo deep thermalization at the subsystem level, offering a complementary and another experimentally accessible route to infer the {resource-generating powers}. Our work connects deep thermalization to resource quantification, offering a new perspective on the essential role of various resources in generating state designs {and on efficient estimation of quantum resources in experiment.}

\end{abstract}

\maketitle

\newtheorem{corollary}{Corollary}[theorem]
\newtheorem{lemma}[theorem]{Lemma}
\newtheorem{remark}{Remark}
\def\endproof{\hfill$\blacksquare$}

\textit{Introduction.---}Quantum resource theories (QRTs) offer a unified framework for characterizing and manipulating operationally useful quantum features.
{Important examples are the QRTs of} entanglement~\cite{amico2008entanglement, RevModPhys.81.865, contreras2019resource, bauml2019resource}, non-stabilizerness~\cite{veitch2014resource, campbell2017roads, leone2024stabilizer}, and coherence~\cite{streltsov2017colloquium, PhysRevLett.119.140402}
{which play central roles in achieving universal quantum computation~\cite{Briegel2009, raussendorf2007topological, gottesman1998heisenberg, aaronson2004improved, bravyi2005universal, bravyi2012magic, hebenstreit2019all}, and asymmetry~\cite{chitambar2019quantum, gour2008resource, marvian2013theory}, which serves as a key resource in the context of quantum reference frames~\cite{chitambar2019quantum}}.
Each QRT specifies free states and operations easily accessible under 
specific experimental constraints, {while the remainder is}, conversely, identified as costly, i.e., resourceful [see Fig.~\ref{fig:sch}(a)]. {For example, in the QRT of entanglement, separable states and LOCCs (local operations and classical communication) are considered free, while entangled states and non-separable operations carry the resource.
The ability of a unitary evolution
to build up the desired resource when starting from a
typical free state is called its ``resource-generating power (RGP)."
It remains an important question \emph{whether and how fast} a unitary, when embedded in a generic quantum circuit, reaches the RGP of a typical, i.e., Haar-random unitary. A comprehensive understanding of this process---in analogy to equilibration in many-body systems~\cite{zhang2015thermalization, jonnadula2017impact, jonnadula2020entanglement, varikuti2025impact} dubbed ``thermalization of quantum resources''---remains an outstanding question.
A closely related---and perhaps even more pressing---challenge concerns their experimental quantification~\cite{elben2023randomized, oliviero2022measuring}.}

{In this work, we build on recent advances in ``deep thermalization''~\cite{choi2023preparing, cotler2023emergent} to engineer a framework that enables (i) the experimental characterization of quantum resources and (ii) provides significant insights into their thermalization in generic quantum circuits \footnote[1]{In this work, we use the term \textit{generic quantum circuits} for those composed of both free and non-free operations.}.
Deep thermalization provides a refined description of statistical relaxation in many-body systems \cite{cotler2023emergent, choi2023preparing, mcginley2023shadow, versini2023efficient, liu2024deep, lucas2023generalized, shrotriya2023nonlocality, mark2024maximum, mok2025optimal, manna2025projected, zhang2025holographic, sherry2025mixed, fritzsch2025free, vairogs2025localizing, mandal2025partial, sreejith2025signatures, liu2025coherence, loio2025quantum, yu2025mixed, yan2025characterizing, chan2024projected, lami2025quantum}, extending beyond the Eigenstate Thermalization Hypothesis \cite{deutsch1991quantum, srednicki1994chaos, ETH_ansatz_expt, d2016quantum, deutsch_18, eth_nonherm}.
Rather than local observables, deep thermalization concerns \textit{projected state ensembles} obtained by post-selecting measurement outcomes on a larger subsystem in a generic non-integrable quantum evolution. These ensembles deeply thermalize, meaning that they relax to 
universal distributions, such as state designs---{finite ensembles of pure states reproducing Haar-random statistics \cite{renes2004symmetric, klappenecker2005mutually, benchmarking2, cotler2023emergent}---or their finite-temperature counterparts, known as Scrooge ensembles~\cite{cotler2023emergent, chang2025deep, manna2025projected}}.}
{It has become a subject of intensive study to analyze these ensembles and their dependency on the initial state \cite{mok2025optimal}, the dynamics \cite{ho2022exact, ippoliti2022solvable, ippoliti2023dynamical, claeys2022emergent, bejan2025matchgate}, or the measurement basis \cite{varikuti2024unraveling, bhore2023deep}.}
{For generic evolutions,
the projected ensembles typically approach state designs, but when both the initial state and dynamics are free in a QRT sense, they deviate significantly from the designs \cite{varikuti2024unraveling, liu2025coherence, 3ttm-vhdt}. This stems from the fact that free elements of QRTs occupy highly structured subsets of the Hilbert space~\footnote[2]{The free unitaries of QRTs considered in this work form proper subgroups of $\mathcal{U}(2^N)$, whose Haar measures generally fail to reproduce the higher-order moments of the Haar measure on $\mathcal{U}(2^N)$, precluding arbitrarily high-order unitary designs.}.
Therefore, the onset of design behavior is intuitively linked to the presence of quantum resources.}
{In this letter, we formalize this intuition by devising a protocol to quantify the RGP of arbitrary evolutions, rooted in the projected ensembles.}
{By particularly focusing on linear entropic resource measures, which are experimentally more accessible than nonlinear ones~\footnote[3]{For example, in the QRT of entanglement, the von Neumann entropy requires diagonalization of the reduced density matrix, whereas the linear entropy, given by the purity, is much easier to compute. From an implementation perspective, the former requires costly operations such as full state tomography, while the latter can be estimated using a SWAP test~\cite{PhysRevLett.120.050406}.}, we construct an unbiased estimator for the RGP.
We find that this estimator is highly reliable, as its signal-to-noise ratio over free elements remains stable with increasing system size. }

{Our protocol can be immediately implemented on near-term quantum devices (see the supplemental material (SM)~\cite{supplemental} for a discussion of experimental feasibility): a free state is evolved under a circuit composed of the target unitary and random free unitaries, followed by projective measurements on a subsystem. The resulting measurement outcome probabilities provide direct access to the RGP of the target unitary. Complementing this, we show that the resource content of the projected states
serves as a proxy for the full dynamics, thereby establishing an experimentally viable framework for probing resources. We focus on asymmetry and non-stabilizerness in the main text, with technical details and the cases of entanglement and coherence presented in the SM~\cite{supplemental}.}

{
\textit{Quantifying resources.---}Many quantum resources admit entropy-based quantifiers, including entanglement~\cite{PhysRevD.82.126010} and stabilizerness~\cite{leone2022stabilizer}.
Let $\mathcal{R}^{(\alpha)}(|\psi\rangle)$ denote the $\alpha$-Rényi entropy quantifying the resource of a state $|\psi\rangle$. The corresponding linear entropy is $\mathcal{R}_{\mathrm{lin}}(|\psi\rangle) = 1 - e^{-\mathcal{R}^{(2)}(|\psi\rangle)}$.
Often, the linear entropies can be evaluated considering a number of $t$ (determined by the underlying QRT) copies of the state of interest, $\mathcal{R}_{\mathrm{lin}}(|\psi\rangle) = 1 - \mathrm{Tr}\left(\mathbf{W} (|\psi\rangle\langle\psi|)^{\otimes t}\right)$. The operator $\mathbf{W}$ is chosen to commute with $F^{\otimes t}$ for all free unitaries $F$ and to satisfy $0\le \langle\psi^{\otimes t}|\mathbf{W}|\psi^{\otimes t}\rangle\le 1$, with the upper bound attained when $|\psi\rangle$ is free, such that $\mathcal{R}_{\mathrm{lin}}$ vanishes for free states.

The resource-generating power (RGP) of a unitary $U$ is then defined as the average linear entropy achieved when $U$ is applied to pure free states. I.e., denoting the average by $\langle\cdot\rangle_{|\psi\rangle}$, $\mathcal{R}_{p}(U) = \langle \mathcal{R}_{\mathrm{lin}}(U|\psi\rangle) \rangle_{|\psi\rangle}$.
RGPs quantify the information-theoretic complexity of implementing a quantum evolution within a QRT framework.
They also bridge QRTs with out-of-time-ordered correlators (OTOCs) and information scrambling~\cite{styliaris2019quantum, varikuti2022out,leone2022stabilizer,anand2021quantum, varikuti2022out}. In addition to characterizing the resource content generated by near-term quantum devices, our results thus provide a scope on experimentally probing quantum chaos.}

\textit{Main results.---}Here, we give a high-level summary of our main results, followed by examples for specific QRTs.
At the heart of our results are the twirling identities, {common to all QRTs considered~\cite{varikuti2025impact, jonnadula2017impact, jonnadula2020entanglement, supplemental}.
They state that averaging the action of a unitary $U$ over free unitaries $F$ and free states $|\psi\rangle$ yields simply a linear interpolation between the $t$'th moments of free ($\rho^{(t)}$) and Haar-random states ($\bm{\Pi}^{(t)}$): $\langle (FU|\psi\rangle\langle\psi|U^{\dagger}F^{\dagger})^{\otimes t}\rangle_{F,|\psi\rangle} = (1-\alpha)\;\rho^{(t)} + \alpha\;\bm{\Pi}^{(t)}$, with $\alpha \propto \mathcal{R}_{p}(U)$
~\footnote[4]{Note that the $t$-th order twirling of a channel $\zeta$ with respect to a compact group $\mathbb{G}(2^N)$ is defined by $\zeta \mapsto \langle F^{\dagger \otimes t}\, \zeta \, F^{\otimes t}\rangle_{F\in \mathbb{G}}$}.}
{Taking advantage of these identities, we obtain a clear measurement protocol for computing RGP in the given QRT}: (i) prepare a random pure bipartite free state $|\psi\rangle_{\mathrm{free}}\in\mathcal{H}_{A}\otimes\mathcal{H}_{B}$, (ii) apply the target unitary $U$, whose RGP is to be estimated, (iii) apply a random free operation $F$, (iv) measure a subset of the final state (say $B$) in the computational basis, record the outcome probabilities $\{p_b\}$, and repeat the above steps several times (see Fig.~\ref{fig:sch}(b)).
{The following theorem establishes one of our main results: an unbiased estimator of the RGP from the probabilities $p_b$.}

\begin{theorem}\label{theorem1}
Let $\mathcal{R}_p(U)$ denote an appropriately defined {RGP} of a unitary $U$ in a {QRT admitting for a twirling identity.}
{Define
$\tilde{P}^{(t)}
=
\sum_{b} p_b^t/2^{N_B},
$
where $t\in \mathbb{Z}^{+}$ is determined by the QRT, and $\{p_b\}$ are the subsystem measurement outcome probabilities} obtained from the above protocol. Then,
\begin{eqnarray}\label{prot-eq}
\langle \mathcal{R}_{\mathrm{est}}(U) \rangle_{F, |\psi\rangle}
=
\frac{\mathcal{R}_p(U)}{\overline{\mathcal{R}}}
=
\frac{k_2-\langle \tilde{P}^{(t)} \rangle_{F,|\psi\rangle}}
{k_2-k_1}.
\end{eqnarray}
Here, $k_1$ and $k_2$ are constants determined by the underlying QRT and the subsystem sizes, $\langle \cdot \rangle_{F,|\psi\rangle}$ denotes averaging over free unitaries $F$ and states $|\psi\rangle$, and $\overline{\mathcal{R}}=\langle \mathcal{R}_p(U)\rangle_{U\in \mathrm{U}(2^N)}$ is the Haar-averaged RGP.
\end{theorem}

\begin{figure}
\includegraphics[scale=0.5]{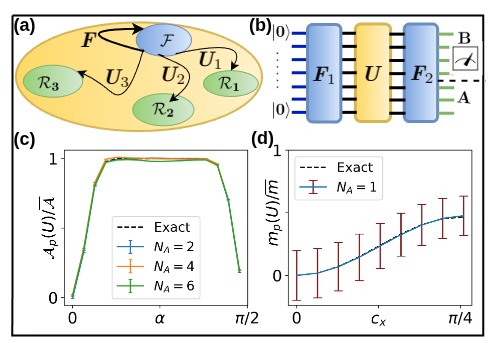}
\caption{\label{fig:sch} (a) Abstract representation of free and non-free elements: $\mathcal{F}$ denotes the set of free states/operators; $U_i$ are non-free operations mapping $\mathcal{F}$ outside itself, while the free operation $F$ preserves $\mathcal{F}$. (b) Protocol for estimating the resource content of a unitary of interest $U$: prepare a random free state $\ket{\psi}$ (e.g., using a free unitary $F_1$ on a fiducial state), apply $U$ followed by a random free operation $F_2$, then perform measurements on a subsystem $B$ and post-process the outcomes. (c) Illustration of the use of the protocol to estimate the $Z_2$-AGP ($\mathcal{A}_{p}(U)/\overline{\mathcal{A}}$) for $U=u^{\otimes N}$ with $u=\exp\{-i\alpha \sigma_{x}\}$, with $N=10$ and subsystem measurements on $N_B=4$, $6$, $8$ qubits. (d) Estimation of non-stabilizing power $m_p(U)/\overline{m}$ for two-qubit unitary $U=\exp\{-ic_x \sigma_x\otimes\sigma_x/2\}$ with projective measurements on $N_B=1$ qubit. }
\end{figure}

The proof proceeds by using the twirling identities to evaluate $\langle p^{t}_b\rangle$, thereby relating {$\langle\tilde{P}^{(t)}\rangle$} to the corresponding RGPs (see Eqs.~(\ref{twirl}) and (\ref{magic_main}), along with the subsequent discussion, and also SM~\cite{supplemental}). {Equation~\eqref{prot-eq} encapsulates the fact that, for computational-basis measurements, $\langle p_b^t\rangle$ is independent of $b$, causing $\langle \tilde{P}^{(t)}\rangle$ to simply linearly interpolate between two regimes: $\langle \tilde{P}^{(t)} \rangle = k_2$ for free $U$, and $\approx k_1$ for Haar-random $U$. }A key advantage of our approach is that measuring only a subsystem suffices, thereby expanding projected ensembles from conceptual relevance into a highly useful tool.

{Importantly, the variance of $\mathcal{R}_{\mathrm{est}}$ over free elements decays at least exponentially with the size of the measured subsystem $N_{B}$ for all the QRTs considered. Remarkably, in the QRT of entanglement, the chaotic limit yields an even stronger suppression, exponentially in $N$, making the protocol particularly attractive (see SM~\cite{supplemental}). Consequently, the signal-to-noise ratio does not deteriorate even for large system sizes.}

{Moreover, the resource content of the projected states is an effective proxy for the full $\mathcal{R}_{p}(U)$.}
Namely, we further show that within the above protocol, {for an arbitrary computational basis state $|b\rangle\in\mathcal{H}_{B}$}, the averaged resource content of the projected states follows to leading order
\begin{eqnarray}\label{2res}
\langle \mathcal{R}_{\text{PS}} \rangle_{F,|\psi\rangle} \approx
\frac{ \Omega\; \overline{\mathcal{R}}_{A} \cdot k_1 }{
\left( 1 - \Omega \right) k_2 + \Omega k_1 }, \;\text{ where }\; \Omega=\dfrac{\mathcal{R}_p(U)}{\overline{\mathcal{R}}}.
\end{eqnarray}
Here, $\overline{\mathcal{R}}_{A}$ denotes the average resource of Haar-random states in the Hilbert space $\mathcal{H}_{A}$.

{The QRT-specific tools we develop to prove the above results also enable a key result on thermalization of quantum resources, namely the equilibration of resource measures to the Haar-ensemble averages in generic quantum circuits. From the perspective of QRTs, generic circuits comprise layers of free and non-free operations, as shown in Fig.~\ref{fig:sch-dt2}(a). By randomizing free operations and fixing the non-free operation to $U$, we show that $\mathcal{R}_{p}$ converges to $\overline{\mathcal{R}}$ exponentially,
a behavior ubiquitous across the QRTs considered:}
\begin{theorem}\label{theorem2}
Let $U^{(n)} = F_{n}U F_{n-1} U \cdots F_1 U$, where each $F_j$ (for all $1 \leq j \leq n$) is drawn independently from the group of free unitaries in a suitable QRT.
Let $U$ be a fixed unitary with finite $\mathcal{R}_{p}(U)$. Then,
\begin{equation}\label{symexp}
\left\langle \mathcal{R}_{p}(U^{(n)}) \right\rangle_{\tilde{F}} = \overline{\mathcal{R}} \left[ 1 - \left( 1 - \frac{\mathcal{R}_{p}(U)}{\overline{\mathcal{R}}} \right)^n \right],
\end{equation}
where $\langle \cdot \rangle_{\tilde{F}}$ denotes independent averaging over the ${F_j}$'s.
\end{theorem}
Interestingly, the same twirling identities that underlie Theorem~\ref{theorem1} also lead to Theorem~\ref{theorem2}. Equation~(\ref{symexp}) {tells us} that the thermalization rate is governed entirely by $\mathcal{R}_{p}(U)$, implying that, with random free operations, even arbitrarily small RGP suffice to generate any targeted amount of resource.
{We use Theorem~\ref{theorem2}, together with Eq.~(\ref{2res}), to also probe the deep thermalization of quantum resources in circuits composed of free and resourceful operations (see Fig.~\ref{fig:sch-dt2}; we define the deep thermalization of resources as the
equilibration of resource measures in projected ensembles of states that undergo generic evolutions.)}
In the following, we rigorously demonstrate the above results for the QRTs of $\mathbb{Z}_2$-asymmetry and non-stabilizerness, {to provide further intuition as well as specific numerical tests.}

\textit{Main results for
$Z_2$-asymmetry.---}{In the QRT of $\mathbb{Z}_2$-asymmetry, free states are given by eigenstates of the symmetry generator $\mathcal{G}=\sigma_z^{\otimes N}$}. We define free unitaries as resource non-increasing operations rather than merely those covariant with the symmetry group, thereby also including unitaries that anti-commute with $\mathcal{G}$ (see SM~\cite{supplemental} for a construction of random unitaries with symmetry via polar decomposition). The linear $Z_2$-asymmetry entropy of a pure state $|\psi\rangle\in\mathcal{H}^{2^N}$ can be computed using~\cite{supplemental}
\begin{equation}
\mathcal{A}_{\mathrm{lin}}(|\psi\rangle) = 1-\text{Tr}\left[\left(\mathbf{Z}^{\otimes 2}_{0}+\mathbf{Z}^{\otimes 2}_{1}\right)(|\psi\rangle\langle\psi )^{\otimes 2}\right],
\end{equation}
where $\mathbf{Z}_{s} = \big(\mathbb{I} + (-1)^s \sigma_z^{\otimes N}\big)/2$ is the projector onto the symmetry sector with charge $s \in \{0,1\}$. $\mathcal{A}_{\mathrm{lin}}$ vanishes iff $|\psi\rangle$ is {an eigenstate of $\mathcal{G}$ and remains invariant under the action of free unitaries}, providing a faithful measure of $Z_2$-asymmetry.
{Correspondingly,} the $Z_2$-asymmetry generating power ($Z_2$-AGP) of a unitary $U$ is
{$\mathcal{A}_{p}(U)=\langle\mathcal{A}_{\mathrm{lin}}(U|\psi\rangle) \rangle_{|\psi\rangle}$.}
{Using elementary algebra, we can prove the following twirling identity~\cite{supplemental}}:
\begin{eqnarray}\label{twirl}
\hspace{-0.75em}\left\langle \left(FU\right)^{\otimes 2}\rho_{\text{f}}  \left(FU\right)^{\dagger\otimes 2} \right\rangle_{F}= \left[ 1-\dfrac{\mathcal{A}_{p}(U)}{\overline{\mathcal{A}}} \right] \rho_{f} +\dfrac{\mathcal{A}_{p}(U)}{\overline{\mathcal{A}}}\bm{\Pi^{(2)}}.
\end{eqnarray}
{Here, $\overline{\mathcal{A}} = 2^{N-1}/(2^N+1)$ is the Haar-averaged AGP \cite{supplemental}, $\rho_{f}$ denotes the second moment of the ensemble of free states, and $\bm{\Pi}^{(2)}=\left(\mathbb{I}+\text{SWAP}\right)/(2^N(2^N+1))$ with $\mathrm{SWAP}$ acting non-trivially on $\mathcal{H}^{2^N}\otimes \mathcal{H}^{2^N}$~\cite{zhang2014matrix}.} As the above equation highlights, while $U$ imprints the resource onto free states, the random $F$ spreads it across the system, yielding a convex combination of free ($\rho_{\mathrm{f}}$) and Haar-randomized ($\bm{\Pi}^{(2)}$) parts, each weighted by terms proportional to $\mathcal{A}_{p}(U)$.

{Following the procedure in Fig.~\ref{fig:sch}(b)}, Eq.~(\ref{twirl}) can be reformulated into an experimentally implementable protocol for measuring $\mathcal{A}_{p}(U)$. Given a random pure free state $|\psi\rangle \in \mathcal{H}_{AB}$ and a random free operation $F$, the probability of projecting on a fixed state $|b\rangle\in\mathcal{H}_{B}$ is $p_b = \text{Tr}(\langle b|FU|\psi\rangle \langle \psi| U^{\dagger}F^{\dagger} |b\rangle)$. By sampling over $|\psi\rangle$ and $F$, and averaging via Eq.~(\ref{twirl}), we obtain $\mathcal{A}_{p}(U)/\overline{\mathcal{A}}=({k_2 - \langle p_b^2 \rangle_{F,|\psi\rangle}})/({k_2 - k_1})$, where $k_1 = \text{Tr}(\langle b^{\otimes 2}|\bm{\Pi}^{(2)}|b^{\otimes 2}\rangle )$, and $ k_2 = \text{Tr}(\langle b^{\otimes 2}|\rho_{\mathrm{f}}|b^{\otimes 2}\rangle )$. This validates the applicability of Theorem~\ref{theorem1} for the QRT of $Z_2$-asymmetry. For computational-basis measurements, {the constants are} $k_{l\in \{1, 2\}}={2^{N_A}(2^{N_A}+l)}/\left(2^N(2^N+l)\right)$ for $0 < N_A < N$ \cite{supplemental}. Therefore, Born probabilities corresponding to subsystem measurements enable a direct evaluation of the AGP.
Interestingly, these probabilities emerge naturally within the projected ensemble framework \cite{cotler2023emergent}.

{Remarkably, in the limit of large $N$, the variance of the estimator obtained from this protocol decays exponentially with $N_{B}$. Specifically, the variance ranges between $\sim 8/(2^{N_{B}}-1)$ to $\sim 2/(2^{N_{B}}-1)$ as $U$ interpolates between a free and a Haar random unitary. For details, see SM~\cite{supplemental}.}
Figure~\ref{fig:sch}(c) shows the {estimation} of the $Z_2$-AGP for an $N=10$ qubit system using projective measurements on $N_B=4, 6, 8$ qubits.

{Similar to the measurement outcome probabilities}, the projected states in this protocol also encode information about $\mathcal{A}_{p}(U)$.
{The projected state on $\mathcal{H}_A$ obtained upon projecting $FU|\psi\rangle$ onto $|b\rangle\in \mathcal{H}_{B}$ is given by}
$|\phi_b\rangle = \langle b|FU|\psi\rangle/\sqrt{p_b}$. By writing $|\tilde{\phi}_b\rangle = \langle b |FU|\psi\rangle$, the asymmetry in the projected state can be computed using $ \mathcal{A}_{\text{PS}}= \text{Tr}[\mathbf{Z}^{(2)\perp}_{A}(|\tilde{\phi}_b\rangle\langle\tilde{\phi}_b |)^{\otimes 2}]/p^2_b$. Then, averaging over free states yields to zeroth order ${\langle\mathcal{A}_{\text{PS}}\rangle}_{F, |\psi\rangle} \approx \left({\mathcal{A}_{p}(U)\; {\overline{\mathcal{A}_{A}}} \; k_1}\right)/\left(\overline{\mathcal{A}}\;\langle p^2_b\rangle_{F, |\psi\rangle}\right)$ (see SM~\cite{supplemental}). {By incorporating $\langle p_b^2\rangle_{F, |\psi\rangle}$ obtained via the protocol, we recover Eq.~(\ref{2res}) for the $Z_2$-AGP. }
When $N_A$ is comparable to $N$, it follows that {$k_1/k_2\approx 1 -2^{-N_A}+2^{-N}$}, leading to $\langle\mathcal{A}_{\text{PS}}\rangle/\overline{\mathcal{A}_{A}} \approx \mathcal{A}_{p}(U)/\overline{\mathcal{A}}$. This indicates that the normalized asymmetry of the projected states, {when averaged over the free states and unitaries}, behaves very similarly to the normalized AGP of $U$.

We now illustrate Theorem~\ref{theorem2}
as a further consequence of the twirling identity in Eq.~(\ref{twirl}).
{The scenario we consider} consists of interlacing free and non-free operations [see Fig.~\ref{fig:sch-dt2}(a)]. The free operations $F_j$ at different time steps $j=1\dots n$ are random and independent, while the non-free operations are arbitrary but taken here as non-random and identical at each step. {The averaged AGP of} $U^{(n)} = F_{n} U F_{n-1} U \cdots F_1 U$ is
\begin{equation}\label{expsymmain}
\left\langle \mathcal{A}_{p}(U^{(n)}) \right\rangle_{\tilde{F}} = \overline{\mathcal{A}} \left[ 1 - \left( 1 - \frac{\mathcal{A}_{p}(U)}{\overline{\mathcal{A}}} \right)^n \right].
\end{equation}
The proof follows from {iteratively applying} the twirling channel in Eq~(\ref{twirl}), {see SM \cite{supplemental}}. This result corresponds to Theorem~\ref{theorem2} for the $Z_2$-AGP. {Notably,} $\overline{\mathcal{A}}$ is the only non-trivial fixed point of Eq.~(\ref{expsymmain}). Moreover, the above equation suggests that when free and non-free operations are applied in conjunction, the AGP typically converges exponentially to its Haar-averaged value with rate $\lambda = -\ln [1 - {\mathcal{A}_{p}(U)}/{\overline{\mathcal{A}}}]$.

\textit{Main results for non-stabilizerness.---}{Non-stabilizerness constitutes a key resource for achieving a quantum advantage, since stabilizer circuits are classically simulable~\cite{veitch2014resource, campbell2017roads, leone2024stabilizer}. In this QRT, the stabilizer states and Clifford group constitute the free states and unitaries.} The non-stabilizerness of a pure state can be computed using the linear stabilizer entropy $
m_{\mathrm{lin}}(|\psi\rangle)=1-2^N\text{Tr}[Q(|\psi\rangle\langle\psi |)^{\otimes 4}],
$
where {$Q=\sum_{i=1}^{2^{2N}}P_{i}^{\otimes 4}/2^{2N}$} is the projector onto the stabilizer space {and $P_i$s are $N$-qubit Pauli strings}~\cite{leone2022stabilizer, leone2021quantum, turkeshi2023measuring, haug2023quantifying, tarabunga2024nonstabilizerness, santra2025complexitytransitionschaoticquantum, varikuti2025impact, capecci2025role, leone2024stabilizer, bittel2025operational, tirrito2024anticoncentration, turkeshi2025magic, tarabunga2025magic, tirrito2025universal, aditya2025mpemba}. {$m_{\mathrm{lin}}$ is a $4$-th degree state polynomial.} The non-stabilizing power of a unitary $U$ is $m_p(U)=1-2^N\text{Tr}[QU^{\otimes 4}\rho_s U^{\dagger\otimes 4}]$, where $\rho_{s}:=\langle\left( |\psi\rangle\langle\psi | \right)^{\otimes 4}\rangle_{|\psi\rangle\in \mathrm{STAB}(N)}$ is the fourth moment of the stabilizer states \cite{leone2022stabilizer, varikuti2025impact, robin2025stabilizer, robin2024magic, chernyshev2025quantum, brokemeier2025quantum}.
Analogous to Eq.~(\ref{twirl}), {algebraic manipulations lead to the twirling identity} (see SM~\cite{supplemental})
\begin{align}\label{magic_main}
\left\langle \left(CU\right)^{\otimes 4}\rho_\text{s} \left( CU \right)^{\dagger \otimes 4} \right\rangle_{C}  = \left[ 1-\dfrac{m_p(U)}{\overline{m}}  \right]\rho_{\text{s}}+\dfrac{m_p(U)}{\overline{m}}\bm{\Pi}^{(4)},
\end{align}
where $\bm{\Pi}^{(4)} = \Pi^{(4)}/\mathrm{Tr}(\Pi^{(4)})$, with $\Pi^{(4)} = \frac{1}{4!}\sum_{\pi \in S_4} \pi$, and {$\pi$ denotes the permutation operator acting on four copies of the Hilbert space $\mathcal{H}^{2^N}$}.
Similar to the $Z_2$-asymmetry, the above equation can be used to devise an experimental protocol to compute $m_p(U)$. Specifically, for a random stabilizer state $|\psi\rangle\in\mathcal{H}_{A}\otimes\mathcal{H}_{B}$ and a random Clifford unitary $C$, the probability to project onto a fixed computational-basis state $|b\rangle \in \mathcal{H}_B$ is related to $m_p(U)$ via Eq.~(\ref{magic_main}) as $m_p(U)/\overline{m} = (k_2 - \langle p_b^4 \rangle_{C, |\psi\rangle}) / (k_2 - k_1)$, where $k_1=\text{Tr}(\langle b^{\otimes 4}| \bm{\Pi}^{(4)} |b^{\otimes 4}\rangle )$ and $k_2=\text{Tr}(\langle b^{\otimes 4}| \rho_s |b^{\otimes 4}\rangle )$. Explicit expressions for $k_1$ and $k_2$ are given in the SM~\cite{supplemental}. This result confirms the applicability of Theorem \ref{theorem1} to the QRT of non-stabilizerness.
{As one can show analytically~\cite{supplemental}, when $N$ is large and $U$ is a Haar-random unitary, the variance of the estimator decays like $\sim 72 /(2^{N_{B}}-1)$.
We anticipate, and corroborate numerically, a similar scaling for the estimator's variance for more general unitaries~\cite{supplemental}.
}

The non-stabilizerness of the projected states can be probed using $m_{\text{PS}}\equiv 1 - 2^{N_A}\text{Tr}[ Q_{A}\left(|\phi_b\rangle\langle\phi_b|\right)^{\otimes 4}]$, where $|\phi_b\rangle=\langle b|CU|\psi\rangle /\sqrt{p_b}$, and $Q_{A}=\sum_{j=0}^{2^{2N_A}} (P_j)^{\otimes 4}/2^{2N_A}$ with $P_j$s being Pauli strings over $N_A$-qubits. The averaged non-stabilizerness of the projected states, up to leading order becomes ${\langle m_{\text{PS}}\rangle} \approx \left({m_p(U)\; k_1}\;{\overline{m_{A}}}\right)/\left(\overline{m}\;\langle p^4_b\rangle\right)$ \cite{supplemental}, where $\langle\cdot\rangle$ indicates averaging over the free elements $C$ and $|\psi\rangle$. This is in accordance with Eq.~(\ref{2res}). When $N_A$ is comparable to $N$, we again have $k_1\approx k_2$, leading to $\langle m_{\text{PS}}\rangle / \overline{m_A} \approx m_p(U) / \overline{m}$. This approximation becomes exact in the no-measurement limit, i.e., $N_A\rightarrow N$. Hence, this result offers another way of characterizing $m_p(U)$. Finally, we remind that the twirling identity in Eq.~(\ref{magic_main}) also allows one to probe the non-stabilizing power of the circuits consisting of alternating free and resourceful unitaries. Let $U^{(n)} = C_{n}U C_{n-1} U \cdots C_1 U$, where each $C_j$ ($\forall\,\,1 \leq j \leq n-1$) is an independently drawn Clifford unitary and $U$ is fixed. Then, Eq.~(\ref{magic_main}) implies \cite{varikuti2025impact}
\begin{eqnarray}\label{symexp_mag}
\left\langle m_p(U^{(n)}) \right\rangle_{\tilde{C}} = \overline{m} \left[ 1 - \left( 1 - \frac{m_p(U)}{\overline{m}} \right)^n \right].
\end{eqnarray}
{As this equation indicates, Theorem~\ref{theorem2} is generic, showing the exponential convergence of quantum resources to their Haar-averaged values when $m_p(U)$ is finite.}

\begin{figure}
\includegraphics[scale=0.5]{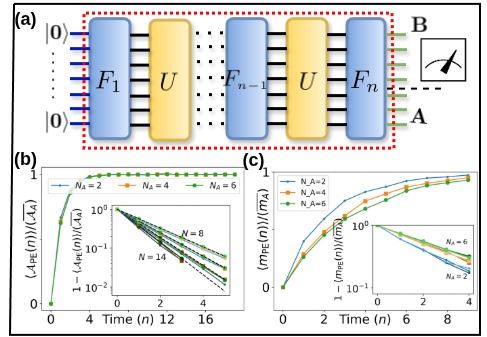}
\caption{\label{fig:sch-dt2} (a) Illustration of repeated free and non-free operations on a fixed free state, followed by projective measurements on a subset of qubits.
(b) Deep thermalization of the $Z_2$-asymmetry for the above circuit with $N = 14$ and subsystem sizes $N_A = 2$ to $6$. The close agreement between the curves shows that subsystem asymmetry provides a reliable probe of the AGP. Inset: relaxation dynamics for $U = u^{\otimes N}$ with $u = e^{-i\pi \sigma_x/24}$ and $N = 8, 10, 12, 14$ (light to dark). For every $N$, the relaxation rates corresponding to different $N_A$ nearly coincide.
(c) Deep thermalization of non-stabilizerness under repeated interspersions of Clifford and non-Clifford operations for the unitary $U = \mathbb{I} \otimes \exp{-i\pi\sigma_x \otimes \sigma_x/8} \otimes \mathbb{I}$, for $N=12$ and different $N_A$. $\langle m(n)_{\mathrm{PE}}\rangle/\overline{m_A}$ approaches $\langle m_p(U^{(n)})\rangle_{\tilde{C}}$ as $N_A\rightarrow N$. Inset: data for $N=8,10,12$ (light to dark). In both (b) and (c), data is averaged over $\sim 10^3$ random circuit instances. }
\end{figure}

\textit{Deep thermalization in QRTs.---}Here, we probe deep thermalization in the QRTs of asymmetry and non-stabilizerness under the circuit dynamics shown in Fig.~\ref{fig:sch-dt2}(a).
Deep thermalization in QRTs can be understood in two equivalent ways: (i) through the thermalization of quantum resources in the projected ensembles, which reflect their subsystem behavior, and (ii) through the emergence of approximate state designs from these ensembles. Both perspectives are intimately related: the thermalization of resources is necessary for the projected ensembles to approximate $t$-designs ($t=2$ and $4$ for asymmetry and non-stabilizerness, respectively). {In the following, we discuss the former perspective, leaving the latter to the End Matter.}

For the $Z_2$-asymmetry, we take the initial states to be random $Z_2$-symmetric. At discrete time $n=0$, measurements in the computational basis ensure that the projected states on the complementary subsystem remain $Z_2$-symmetric, and therefore $\langle \mathcal{A}_{\text{PE}}\rangle_{|\psi\rangle, \tilde{F}} = 0$, {where $\mathcal{A}_{\mathrm{PE}}=\sum_{b}p_b \mathcal{A}_{\mathrm{PS}}(|\phi_b\rangle)$ denotes the resource content of the projected ensemble for a single circuit instance}. As the states evolve under the circuit dynamics, their global asymmetry grows and relaxes to the Haar-averaged value in accordance with Eq.~(\ref{expsymmain}). This indicates that the states become Haar-like at the level of their second moment. Since ensembles of $2$-design states yield projected ensembles that also approximate $2$-designs \cite{cotler2023emergent}, the asymmetry of the projected ensembles likewise grows and thermalizes to the Haar value, which we identify as deep thermalization of $Z_2$-asymmetry.
{Equation~(\ref{2res}) further shows that the asymmetry of the projected states closely tracks the AGP of the full dynamics.}
In particular, for $N_A\sim N$, we have $\langle \mathcal{A}_{\text{PS}}(n)\rangle_{|\psi\rangle, \tilde{F}} /\overline{\mathcal{A}}_{A} \approx \langle \mathcal{A}_p(U^{(n)})\rangle_{\tilde{F}}/\overline{\mathcal{A}}$. Consequently, the collective asymmetry of the projected ensemble can be expected to follow a similar exponential relaxation as in Eq.~(\ref{expsymmain}), as confirmed numerically in Fig.~\ref{fig:sch-dt2}(b). The figure shows $\langle \mathcal{A}_{\text{PE}}(n)\rangle_{\tilde{F}}/\overline{\mathcal{A}}_{A}$ for $N=8, 10, 12, 14$, with $N_A=2, 4, 6$, all converging to the normalized AGP $\langle \mathcal{A}_{p}(U^{(n)})\rangle_{\tilde{F}}/\overline{\mathcal{A}}$, indicating the deep thermalization of the asymmetry. The inset compares $\langle \mathcal{A}_{\text{PE}}(n)\rangle_{\tilde{F}}$ with the AGP of the dynamics, $\langle \mathcal{A}(U^{(n)})\rangle_{|\psi\rangle, \tilde{F}}$, showing remarkable agreement even for small subsystems ($N_A=2$), consistent with our analytical predictions.

{
We now turn to the behavior of non-stabilizerness in projected ensembles, given by $m_{\mathrm{PE}} = \sum_{b}p_b m_{\mathrm{PS}}(|\phi_{b}\rangle)$, under interlaced Clifford and non-Clifford dynamics. At $n=0$, computational-basis measurements leave the projected states in the stabilizer set, yielding $m_{\mathrm{PE}} = 0$. As the dynamics proceeds, the non-stabilizerness of the projected ensemble, averaged over initial stabilizer states and Clifford unitaries, grows and relaxes exponentially toward the Haar value $\bar{m}_A$ for all system and subsystem sizes considered [Fig.~\ref{fig:sch-dt2}(c)]. For larger subsystems, $\langle m_{\mathrm{PE}}(n)\rangle/\bar{m}_A$ closely tracks the non-stabilizing power of the full dynamics. The inset shows even faster relaxation of $\langle m_{\mathrm{PE}}(n)\rangle$ than the global evolution. This numerical behavior is consistent with Eq.~(\ref{symexp_mag}).
}

\textit{Discussion.---}{In this Letter, we have developed a protocol to quantify and characterize the resource-generating power of unitaries across a broad class of resource theories. The quantum resources the framework encompasses have implications for quantum computation, quantum foundations, and the classical simulability of quantum dynamics.
A central outcome is an unbiased estimator whose variance over free elements decays at least exponentially with the size of the measured subsystem, enabling scalable experimental characterization of quantum resources. Beyond its practical utility, our framework reveals that free and resourceful operations together drive the thermalization of resource measures toward their Haar values, while deep thermalization provides a complementary subsystem-level proxy for the resource content of dynamics.
}

{Being built on twirling identities, our framework naturally extends to quantifying the resource content of arbitrary quantum states.}
Although our protocol relies on random free operations, it remains an open question whether similar twirling identities emerge when using local free operations \cite{oliviero2022measuring, coffman2024local, elben2023randomized},
{which are often more scalable in practice.} {In addition, developing efficient methods for sampling free unitaries constitutes an important direction for future research.}
Our main results hinge on twirling identities like Eqs.~(\ref{twirl}) and~(\ref{magic_main}) (see also SM \cite{supplemental}), where averaging free unitaries yields a weighted sum of free and Haar moments. Any QRT admitting such a twirling identity will lead to the same conclusions.
Investigating whether this framework extends to other QRTs, such as non-Gaussianity \cite{jozsa2008matchgates, hebenstreit2019all, knill2001fermionic, coffman2025measuring, braccia2024computing, 3yx4-1j27, w97w-7zny}, or non-stabilizerness for qudits~\cite{magni2025anticoncentration, Magni2025quantumcomplexity}, is a compelling direction of future research.

\textbf{Acknowledgments.} NDV and SB acknowledge insightful discussions with Arul Lakshminarayan. We thank Andrea Legramandi and Leela Ganesh Chandra Lakkaraju for discussions related to quantum resources.
This project has received funding from the Italian Ministry of University and Research (MUR) through project DYNAMITE QUANTERA2\_00056, in the frame of  ERANET COFUND QuantERA II – 2021 call co-funded by the European Union (H2020, GA No 101017733); the European Union - Next Generation EU, Mission 4, Component 2 - CUP E53D23002240006, and CARITRO through project SQuaSH.
This project was supported by the European Union under Horizon Europe Programme - Grant Agreement 101080086 - NeQST; the Swiss State Secretariat for Education, Research and lnnovation (SERI) under contract number UeMO19-5.1; the Provincia Autonoma di Trento, and Q@TN, the joint lab between University of Trento, FBK—Fondazione Bruno Kessler, INFN—National Institute for Nuclear Physics, and CNR—National Research Council.
S.B.\ acknowledges CINECA for use of HPC resources
under Italian SuperComputing Resource Allocation– ISCRA
 Class C Project No.DeepSYK - HP10CAD1L3.
Views and opinions expressed are however those of the author(s) only and do not necessarily reflect those of the European Union or of the Ministry of University and Research.
Neither the European Union nor the granting authority can be held responsible for them.

\textbf{Data Availability Statement.} All data supporting the findings of this study are publicly available in the Zenodo repository at~\cite{varikuti2026deepdata}

\bibliography{main.bbl}
\section*{End Matter}

\begin{figure}[t]
\includegraphics[scale=0.31]{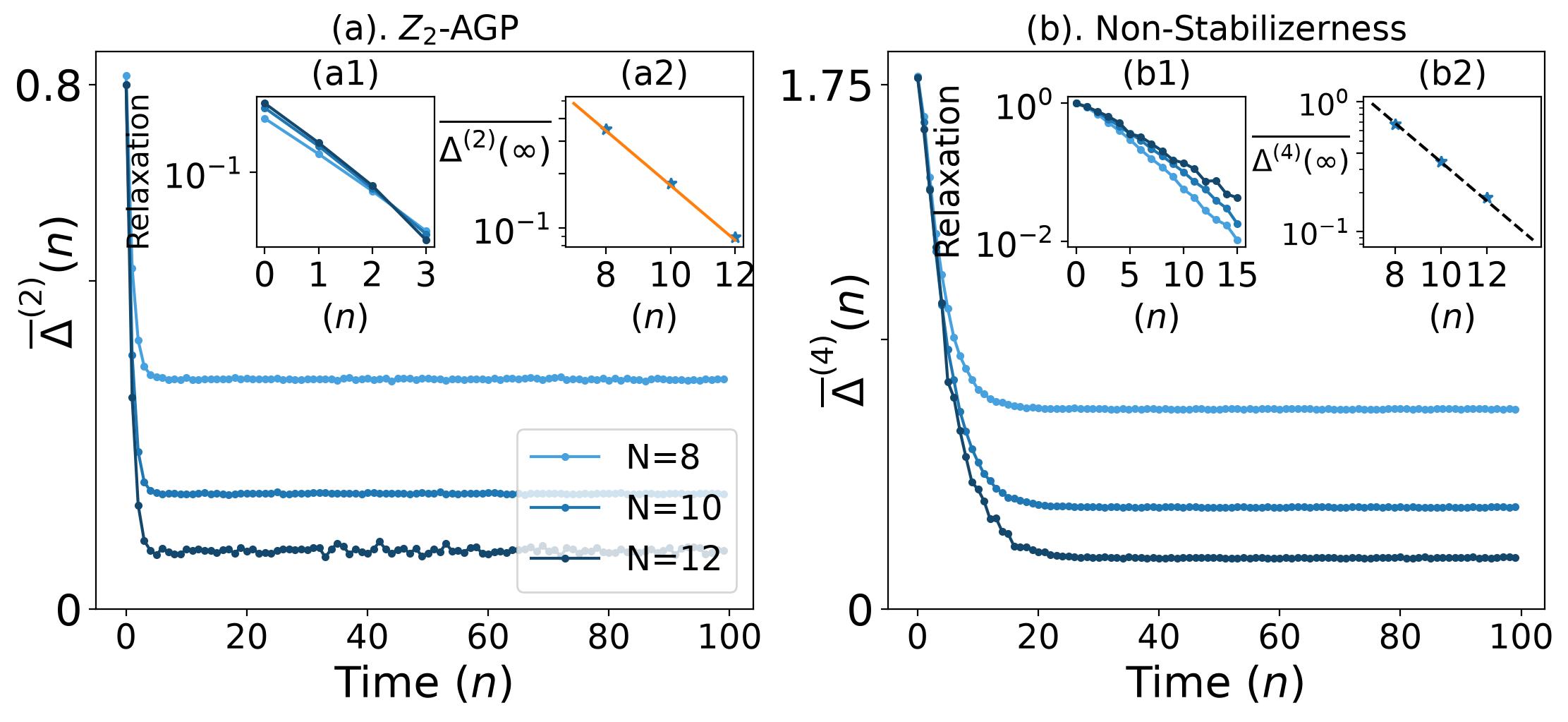}
\caption{\label{fig:td} Deep thermalization, quantified by the trace distance between the moments of the projected ensemble and that of the Haar ensemble, in the QRT of (a) $Z_2$-asymmetry and (b) non-stabilizerness. In both cases, the initial free state evolves under repeated interspersings of random free and non-free operations as depicted in Fig.~\ref{fig:sch-dt2}(a), followed by measurements in the computational basis. The data are averaged over $\sim 10^2$ realizations of the circuit and the initial free state. The subsystem size is kept fixed at $N_A = 2$. Insets: (a1,b1) The trace distance exhibits exponential relaxation toward its long-time average across all system sizes, consistent with global thermalization. We quantify this relaxation using $1 - \overline{\Delta}^{(t)} / \overline{\Delta^{(t)}(\infty)}$,
where $\overline{\Delta^{(t)}(\infty)}$ denotes the long-time average and $t = 2,4$ correspond to the $\mathbb{Z}_2$-AGP and non-stabilizerness, respectively. (b1,b2) The long-time average decays exponentially with increasing system size.  }
\end{figure}

\textit{Emergent state designs in QRTs.---}Complementary to the analysis in the main text regarding the thermalization of quantum resources in projected ensembles, we further characterize deep thermalization by studying how the projected ensembles of typical free generator states converge to state $t$-designs \footnote[5]{Recall that $t=2$ and $4$ for the QRTs of $Z_2$-asymmetry and non-stabilizerness, respectively.}. For this purpose, we use the trace-distance measure between the moments, given by $\Delta^{(t)}(\tau) = \| \mathcal{M}^{(t)}_{\text{PE}}(n) - \Pi^{(t)}_{A} \|_{1}$. At $n=0$, the first moment of the projected ensemble already coincides with the Haar moment ($\mathbb{I}{A}/2^{N_A}$), while higher moments remain symmetrized. Under the dynamics generated by $U^{(n)}$ in Fig.~\ref{fig:sch-dt2}(a), these moments relax exponentially in time toward their Haar counterparts. Figures~\ref{fig:td}(a) and \ref{fig:td}(b) show this relaxation for the second and fourth moments, respectively, for $N=8,10,12$ with $N_A=2$, where measurements on $N_B$ qubits are performed in the computational basis and results are averaged over $\sim 10^2$ circuit realizations. The second-moment distance $\Delta^{(2)}(n)$ decays exponentially with time, and its long-time value decreases exponentially with increasing $N_B$.
There are well-known situations where first non-trivial deviations from Haar statistics occur at higher moments (e.g., stabilizer states without projections form exact 3-designs, but not 4-designs \cite{zhu2016clifford, leone2025non}).
Therefore, we also study $\Delta^{(4)}(n)$, which again shows exponential decay, indicating rapid convergence of the projected ensemble toward Haar-like statistics, with long-time values that vanish exponentially with system size.

\onecolumngrid

\vspace{2em}

\begin{center}
\textbf{\large Supplemental Material: Deep Thermalization and Measurements of Quantum Resources}
\end{center}

\setcounter{theorem}{0}
\renewcommand{\thetheorem}{A\arabic{theorem}}

\tableofcontents
\newpage

\section{Details on resource quantifiers}
\label{appa}

In this section, we introduce the entropic resource quantifiers associated with the two quantum resource theories (QRTs) considered in the main text, namely, the QRT of $\mathbb{Z}_2$-asymmetry and the QRT of non-stabilizerness. For the QRT of $\mathbb{Z}_2$-asymmetry, we provide explicit details concerning the structure of the free states and the corresponding free operations, thereby clarifying how the entropic measures arise in this setting. For the QRT of non-stabilizerness, the framework is comparatively well established in the literature, and hence, we focus on presenting the entropic quantifier relevant to this resource. Taken together, these discussions serve to highlight how entropy-based measures can provide a unifying lens through which different resource theories may be compared and understood.

\subsubsection{$Z_2$-Asymmetry generating power}
Here, we provide useful details concerning the quantum resource theory of $Z_2$-asymmetry \cite{chitambar2019quantum, gour2008resource, marvian2013theory}. The $Z_2$-symmetry has a single generator $\sigma^{\otimes N}_z$ supported over an $N$-qubit Hilbert space. Then, the Hilbert space can be decomposed into two invariant sectors having the eigenvalues or charges $\pm 1$, with corresponding subspace projectors
\begin{eqnarray}
 \mathbf{Z}_{0} = \dfrac{\mathbb{I}_{2^N}+\sigma^{\otimes N}_z}{2}   \;\text{ and }\;   \mathbf{Z}_{1} = \dfrac{\mathbb{I}_{2^N}-\sigma^{\otimes N}_z}{2}.
\end{eqnarray}
\textbf{Free states and operations: }The eigenvectors of $\sigma^{\otimes N}_z$, whose ensemble we denote with $\mathcal{E}_{|\psi\rangle_{\text{free}}}$, constitute the free states of the $Z_2$-asymmetry. We define free operations as resource-non-generating maps that map free states to free states. This broader class includes unitaries that permute sectors or basis states without generating the resource. Importantly, it avoids the extra assumptions needed under conventional covariant operations, while still remaining consistent with the overall resource-theory framework \footnote[4]{The results obtained here for $Z_{2}$-asymmetry may also extend, under mild additional assumptions, to the setting where the free operations are covariant.}. Under this definition, the free unitaries of the $Z_2$-asymmetry QRT constitute the union of two disjoint sets:
\begin{enumerate}
    \item  The set of unitaries that commute with $\sigma^{\otimes N}_z$, i.e., $\{u\in \text{U}(2^N)\; |\; [u, \sigma^{\otimes N}_{z}]=\mathbf{0} \}$, denoted with $\mathcal{U}_{C}$ (see next Section for a simple construction of these unitaries using polar decomposition), and
    \item all the unitaries that anti-commute with $\sigma^{\otimes N}_z$, i.e., $\{ u\in \text{U}(2^N)\;|\; \{u, \sigma^{\otimes N}_z\}=\mathbf{0} \}$, denoted with $\mathcal{U}_{AC}$.
\end{enumerate}
When the free states are acted upon by the unitaries from $\mathcal{U}_{C}$, their charge remains preserved, i.e., for some free state $|\psi_k\rangle$ with charge $k$, we have $u|\psi_k\rangle = (-1)^k|\psi_k\rangle$ for every $u\in \mathcal{U}_{C}$. However, it gets flipped under the application of unitaries from $\mathcal{U}_{AC}$ --- for any $u\in \mathcal{U}_{AC}$, we have $u\;|\psi_k\rangle = |\phi_{k+1 \;\text{mod}\; 2}\rangle$, where $|\psi\rangle$ and $|\phi\rangle$ are two arbitrary free states with different charges. Non-free operations, in contrast, create superpositions of the free states from different charge sectors. With a slight rearrangement of the basis, one can write the unitaries belonging to $\mathcal{U}_{C}$ and $\mathcal{U}_{AC}$ as
\begin{eqnarray}
u_1\equiv \raisebox{-1.2cm}{\includegraphics[scale=0.15]{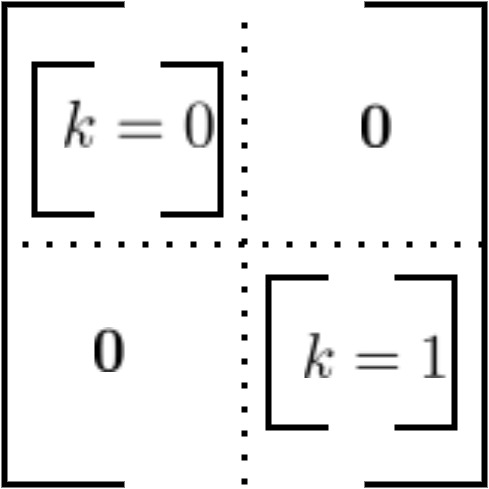}}   \;\;\;\text{ and }\;\;\; u_2\equiv \raisebox{-1.2cm}{\includegraphics[scale=0.15]{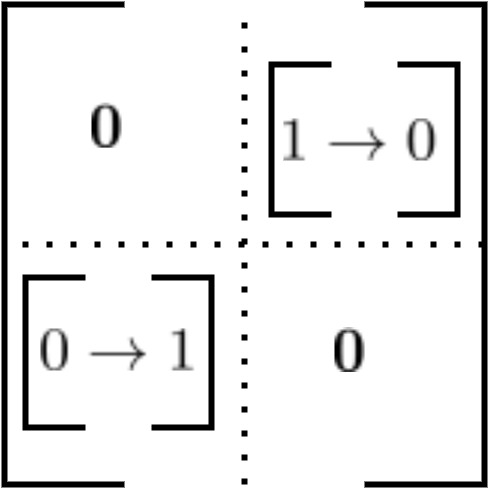}}.
\end{eqnarray}
It is then easy to see that $u_2$ can be obtained by permuting the blocks of $u_1$, i.e., $u_2= P u_1$, where $P$ transforms the basis vectors from one charge sector to the other. Then, we have $P\mathbf{Z}_kP^{\dagger}= \mathbf{Z}_{k+1\;\text{mod }2}$.

We now show that $\mathcal{U}_{\mathbb{Z}_2} \equiv \mathcal{U}_{C} \cup \mathcal{U}_{AC}$ is a subgroup of the unitary group $\mathrm{U}(2^N)$. First, note that if $w_1, w_2 \in \mathcal{U}_{C}$, then $w_1 w_2 \in \mathcal{U}_{C}$. Similarly, if $w_1, w_2 \in \mathcal{U}_{AC}$, then $w_1 w_2 \in \mathcal{U}_{C}$. Finally, if $w_1 \in \mathcal{U}_{C}$ and $w_2 \in \mathcal{U}_{AC}$, then $w_1 w_2 \in \mathcal{U}_{AC}$. Thus, $\mathcal{U}_{\mathbb{Z}_2}$ is closed under multiplication. Next, consider $u \in \mathcal{U}_{C}$. We can write $u = u_{k=0} \oplus u_{k=1}$, hence its inverse is $u^{\dagger} = u^{\dagger}_{k=0} \oplus u^{\dagger}_{k=1}$, which again lies in $\mathcal{U}_{C}$. Similarly, the inverse of any unitary in $\mathcal{U}_{AC}$ also belongs to $\mathcal{U}_{AC}$. Therefore, $\mathcal{U}_{\mathbb{Z}_2}$ is closed under taking inverses. Since $\mathcal{U}_{\mathbb{Z}_2}$ is closed under both multiplication and inversion, it follows that $\mathcal{U}_{\mathbb{Z}_2}$ is a subgroup of $\mathrm{U}(2^N)$. In the following, we show that $\mathcal{U}_{\mathbb{Z}_2}$ is a compact subgroup of the unitary group $\mathcal{U}(2^N)$, thereby showing that it can be associated with a natural Haar measure.

\begin{lemma}
Let $\mathcal{U}_{\mathbb{Z}_2}(2^N)$ denote the subset of $\mathcal{U}(2^N)$ consisting of all unitaries that either commute or anti-commute with the $\mathbb{Z}_2$ symmetry operator $S=\sigma^{\otimes N}_{z}$, i.e.,
$[w, S] = 0\; \text{ or }\; \{w, S\}=0\quad \text{for any} \quad w \in \mathcal{U}_{\mathbb{Z}_2}(2^N)$. Then, $\mathcal{U}_{\mathbb{Z}_2}(2^N)$ forms a compact subgroup of the full unitary group $\mathcal{U}(2^N)$.
\end{lemma}
\begin{proof}
We have already shown that $\mathcal{U}_{\mathbb{Z}_2}(2^N)\equiv \mathcal{U}_{C}\cup\mathcal{U}_{AC}$ is a subgroup of $\mathcal{U}(2^N)$, since it is a subset of the latter and satisfies the group axioms. Because the operator norm of unitary matrices is bounded, $\mathcal{U}_{\mathbb{Z}_2}(2^N)$ is bounded as well. To establish closedness, note that the union of two closed sets is always closed. In particular, $ \mathcal{U}_{C} = \{ w \in \mathcal{U}(2^N) : S^\dagger w S - w = 0 \}$ is the preimage of the zero matrix under the continuous map $w \mapsto S^\dagger w S - w$, while $\mathcal{U}_{AC} = \{ w \in \mathcal{U}(2^N) : S^\dagger w S + w = 0 \}$ is the preimage of the zero matrix under the continuous map $w \mapsto S^\dagger w S + w$. Hence, both $\mathcal{U}_{C}$ and $\mathcal{U}_{AC}$ are closed, and so is their union $\mathcal{U}_{\mathbb{Z}_2}(2^N)$. Since $\mathcal{U}_{\mathbb{Z}_2}(2^N)$ is both closed and bounded in the finite-dimensional normed space $M_{2^N}(\mathbb{C})$, it is compact. Therefore, $\mathcal{U}_{\mathbb{Z}_2}(2^N)$ is a compact subgroup of $\mathcal{U}(2^N)$, and consequently admits a natural Haar measure.
\end{proof}

We shall make use of the above result to perform Haar integrals over the free unitaries in the QRT of $Z_2$-asymmetry. Now, we provide the quantifiers of the $Z_2$-asymmetry in quantum states and the quantum evolutions.

\textbf{$\alpha$-R\'enyi entropy of $Z_2$-asymmetry: }Given a quantum state $|\psi\rangle\in\mathcal{H}^{2^N}$, we define the probability vector $\mathbb{P}\equiv \{\langle \psi |\mathbf{Z}_0|\psi\rangle , \langle \psi |\mathbf{Z}_1|\psi\rangle\}$. When $|\psi\rangle$ is an eigenstate of the symmetry operator $\sigma^{\otimes N}_z$, one of the two elements in $\mathbb{P}$ becomes $1$, while the other vanishes. One can then define the $\alpha$-R\'enyi-entropy, as a measure of the $Z_2$-asymmetry:
\begin{eqnarray}
\mathcal{A}^{\alpha}(|\psi\rangle) = \dfrac{1}{1-\alpha}\ln\left( \langle \mathbf{Z_0}\rangle^{\alpha} + \langle \mathbf{Z_1}\rangle^{\alpha} \right).
\end{eqnarray}
The above quantity vanishes if and only if $|\psi\rangle$ is an eigenvector of the symmetry operator. Also, under the action of free operations mentioned above, $\mathcal{A}^{\alpha}(|\psi\rangle)$ remains invariant. Therefore, $\mathcal{A}^{\alpha}$  provides a faithful measure of the $Z_2$-asymmetry of quantum states. In this work, we are interested in linear asymmetry, given by
\begin{eqnarray}
\mathcal{A}_{\text{lin}}(|\psi\rangle)=\text{Tr}\left[ \mathbf{Z}^{(2)\perp}\left( |\psi\rangle\langle\psi | \right)^{\otimes 2} \right] \,,
\end{eqnarray}
where $\mathbf{Z}^{(2)\perp}=\mathbb{I}-\mathbf{Z}^{\otimes 2}_{0}-\mathbf{Z}^{\otimes 2}_{1}$ denotes the complement of the $Z_2$-symmetric subspaces over the two-replica space.

\textbf{Asymmetry Generating Power: }Having defined all the necessary tools, one can define the asymmetry generating power (AGP) of an arbitrary unitary $U$ as the average amount of symmetry breaking that $U$ introduces as it acts on a typical state with the $Z_2$-symmetry:
\begin{eqnarray}
\mathcal{A}_{p}(U)= \text{Tr}\left[\mathbf{Z}^{(2)\perp}U^{\otimes 2}\overline{\left( |\psi\rangle\langle\psi | \right)^{\otimes 2}} U^{\dagger \otimes 2} \right],
\end{eqnarray}
where the overline indicates the statistical average over all free states from $\mathcal{E}_{|\psi\rangle_{\text{free}}}$. This average can be evaluated as
\begin{eqnarray}
\overline{\left( |\psi\rangle\langle\psi | \right)^{\otimes 2}} = \dfrac{1}{2}\left[   \dfrac{\mathbf{Z}^{\otimes 2}_0 \Pi^{(2)}}{\text{Tr}\left( \mathbf{Z}^{\otimes 2}_0 \Pi^{(2)} \right)} + \dfrac{\mathbf{Z}^{\otimes 2}_1 \Pi^{(2)}}{\text{Tr}\left( \mathbf{Z}^{\otimes 2}_1 \Pi^{(2)} \right)}  \right].
\end{eqnarray}
{Here, $\Pi^{(2)} = ({\mathbb{I} + \mathrm{SWAP}})/{2}$ denotes the projector onto the permutation-invariant subspace of the two-copy Hilbert space $\mathcal{H}^{2^N} \otimes \mathcal{H}^{2^N}$, where $\mathrm{SWAP}$ exchanges the two replicas}. The factor $1/2$ arises from the fact that both charge sectors have equal dimension and thus equal probability of being chosen. Also, note that $\text{Tr}\left( \mathbf{Z}^{\otimes 2}_0 \Pi^{(2)} \right) = \text{Tr}\left( \mathbf{Z}^{\otimes 2}_1 \Pi^{(2)} \right)$. It then follows that
\begin{eqnarray}
\mathcal{A}_{p}(U)&=& 1- \dfrac{1}{2\mathcal{Z}}\sum_{s_1, s_2\in\{0, 1\}}\text{Tr}\left(  \mathbf{Z}^{\otimes 2}_{s_1}U^{\otimes 2}\mathbf{Z}^{\otimes 2}_{s_2}\Pi^{(2)} U^{\dagger\otimes 2} \right)      \nonumber\\
&=& 1- \dfrac{1}{4\mathcal{Z}}\sum_{s_1, s_2\in\{0, 1\}}\left[\text{Tr}\left(  \mathbf{Z}_{s_1}U\mathbf{Z}_{s_2}U^{\dagger} \right)^2 + \text{Tr}\left(  \mathbf{Z}_{s_1}U\mathbf{Z}_{s_2}U^{\dagger} \mathbf{Z}_{s_1}U\mathbf{Z}_{s_2}U^{\dagger} \right) \right],
\end{eqnarray}
where $\mathcal{Z}=\text{Tr}(\mathbf{Z}^{\otimes 2}_{0}\Pi^{(2)})=\text{Tr}(\mathbf{Z}^{\otimes 2}_{1}\Pi^{(2)})$. The above expression implies that the linear AGP of the arbitrary unitary $U$ is directly related to the two-point and four-point out-of-time ordered correlators (OTOCs) of $\mathbf{Z}_{0}$ and $\mathbf{Z}_{1}$. Similar connections between the OTOCs and other linear resource-generating powers, such as non-stabilizerness and entanglement, have been found in the literature \cite{leone2022stabilizer, styliaris2021information, varikuti2022out}. Interestingly, the OTOCs were shown to have surprising connections with asymmetry through the Fisher information, as noted in Ref.~\cite {garttner2018relating}. Moreover, the AGP of the Haar random unitaries is
\begin{equation}
\overline{\mathcal{A}}=1-\dfrac{1}{\mathrm{Tr}(\Pi^{(2)})}\left[\text{Tr}(\mathbf{Z}^{\otimes 2}_{0}\Pi^{(2)})+ \text{Tr}(\mathbf{Z}^{\otimes 2}_{1}\Pi^{(2)})\right] = \dfrac{2^N}{2(2^N+1)}.
\end{equation}
In the large-$N$ limit, $\overline{\mathcal{A}}$ approaches $1/2$. Moreover, it is worth noting that both the $Z_2$ asymmetry of a quantum state and the $Z_2$-AGP of a unitary operator are always upper bounded by $1/2$.

\begin{table*}[t]
\centering
\caption{Summary of the quantum resource theories considered in this work.}
\label{tab:resource_theories}
\renewcommand{\arraystretch}{1.}
\begin{tabular}{|l|l|l|l|}
\hline
\parbox[c]{2.5cm}{\textbf{Resource Theory}} &
\parbox[c]{3cm}{\textbf{Free (Pure) States}} &
\parbox[c]{4.5cm}{\textbf{Free Unitaries}} &
\parbox[c]{5cm}{\textbf{Resource Quantifier}} \\
\hline
\parbox[c]{2.5cm}{
Entanglement} &
\parbox[c]{3.75cm}{Product states $\ket{\psi_A}\otimes\ket{\psi_B}$} &
\parbox[c]{5cm}{
Local unitaries $U_A\otimes U_B$
}
&
\parbox[c]{5cm}{
\vspace{0.08in}
Linear entanglement entropy:
\[
E_{\rm lin}(\psi)=1-\mathrm{Tr}(\rho_A^2)
\]
}
\\
\hline
\parbox[c]{2.5cm}{
Coherence} &
\parbox[c]{3.5cm}{Computational basis states} &
\parbox[c]{5.25cm}{
\[
U=\sum_i e^{i\theta_i}|\pi(i)\rangle\langle i|\,,
\]
\[
\text{with }\pi \text{ a basis vector permu-}
\]
\[
\text{tation operator, }
\pi |i\rangle = |\pi(i)\rangle
\]
}
&
\parbox[c]{5cm}{
Linear coherence:
\[
C_{\rm lin}(\psi)=1-\sum_i|\langle i|\psi\rangle|^4
\]
}
\\
\hline

\parbox[c]{2.5cm}{$\mathbb{Z}_2$-asymmetry} &
\parbox[c]{3.5cm}{Eigenstates of $\sigma^{\otimes N}_z$} &
\parbox[c]{5cm}{
\[
U\in\mathcal{H}^{2^N} \text{ satisfying}
\]
\[
[U,\sigma_z^{\otimes N}]=0
\quad\text{or}\quad
\{U,\sigma_z^{\otimes N}\}=0
\]
}
&
\parbox[c]{5.3cm}{
Linear asymmetry:
\[
A_{\rm lin}(\psi)
=
\mathrm{Tr}\!\left[
Z_{\perp}^{(2)}
(\ket{\psi}\!\bra{\psi})^{\otimes2}
\right]
\]
}
\\
\hline

\parbox[c]{2.5cm}{
Non-stabilizerness\\
(Magic)
}
&
\parbox[c]{3.5cm}{Stabilizer states} &
\parbox[c]{5cm}{
Clifford unitaries:
\[
\text{generated by }\{\text{CNOT, H, S} \} \text{ gates}.
\]
}
&
\parbox[c]{5.5cm}{
\[
\text{Linear stabilizer entropy:}
\]
\[
m(\psi)
=
1-
2^N
\mathrm{Tr}\!\left[
Q(\ket{\psi}\!\bra{\psi})^{\otimes4}
\right]
\]
}
\\
\hline

\end{tabular}
\end{table*}

\subsubsection{Non-stabilizing power}
Here, we briefly summarize the definitions of the non-stabilizerness of a state and the non-stabilizing power of a unitary. A pure state in an $N$-qubit Hilbert space is called a stabilizer state if it is stabilized by a subgroup $\mathcal{S} \subset \mathcal{G}_N$ of Pauli strings with $|\mathcal{S}| = 2^N$, making it a simultaneous $+1$ eigenstate of every $P \in \mathcal{S}$. Such states can be efficiently generated using Clifford circuits and are therefore classically simulable. Given a state $|\psi\rangle$, one can quantify how far it is from the stabilizer space using the linear stabilizer entropy given by \cite{leone2022stabilizer}
\begin{eqnarray}
m(|\psi\rangle) = 1 - 2^N \sum_{i=1}^{2^{2N}} \frac{1}{2^{2N}} \langle \psi | P_i | \psi \rangle^4 = 1 - 2^N \text{Tr}[Q(|\psi\rangle \langle \psi|)^{\otimes 4}],
\end{eqnarray}
where $\{P_i\}$ are all $N$-qubit Pauli strings and $Q$ is a projector in $\mathcal{H}^{\otimes 4}$ and is defined as $Q=\sum_{i=1}^{2^{2N}}P^{\otimes 4}_i/2^{2N}$. Note that $m(|\psi\rangle)$ is invariant under Clifford operations and vanishes if and only if $|\psi\rangle$ is a stabilizer state. Therefore, the stabilizer entropy provides a faithful measure of non-stabilizerness of quantum states. Recently, this quantity has been shown to be a strong monotone of non-stabilizerness \cite{leone2024stabilizer}. Following this, the linear non-stabilizing power of an arbitrary quantum evolution can be defined as
\begin{eqnarray}
m_p(U) = 1 - 2^N \text{Tr}\left[ Q U^{\otimes 4} \overline{(|\psi\rangle \langle \psi|)^{\otimes 4}} U^{\dagger \otimes 4} \right],
\end{eqnarray}
where the overline denotes averaging over all stabilizer states. The non-stabilizing power of a unitary quantifies the average amount of non-stabilizerness that the unitary generates when it acts upon a typical stabilizer state. Similar to the stabilizer entropy, this quantity is also invariant under Clifford conjugation and vanishes if and only if $U$ is a Clifford operator.

The linear stabilizer entropy connects naturally with quantum information scrambling diagnostics, and we focus on it here due to its operational meaning and tractable analytic structure \cite{leone2021quantum, leone2022stabilizer, bittel2025operational}. Other measures like stabilizer fidelity, extent \cite{bravyi2019simulation}, rank \cite{bravyi2019simulation, bravyi2016trading, aaronson2004improved}, mana \cite{veitch2014resource}, and Wigner negativity \cite{pashayan2015estimating} have also been explored in related contexts.

\vspace{1em}
{In the main text, we primarily considered the aforementioned quantum resources to present our main results; see Table~\ref{tab:resource_theories} for a summary of the resource theories studied in this work. Our framework relies on the free unitaries of the underlying resource theory forming a compact group, which enables the use of Haar integral techniques in our derivations. Extending these results to other resource theories, such as quantum thermodynamics or fermionic non-Gaussianity, where similar group-theoretic structures may exist, remains an interesting direction for future work.}

\section{Projected state ensembles}
Here, we briefly summarize the projected state ensembles introduced in Ref.~\cite{cotler2023emergent}. A quantum state $t$-design is an ensemble of pure states whose $t$-th moments match those of the Haar distribution. Formally, an ensemble $\mathcal{E} = \{p_i, |\psi_i\rangle\}$ is an exact $t$-design if for all $k\leq t$, the following holds \cite{renes2004symmetric, klappenecker2005mutually, benchmarking2}:
\begin{equation}
\sum^{|\mathcal{E}|}_{i=1} p_i (|\psi_i\rangle\langle\psi_i|)^{\otimes k} = \int d\psi\, (|\psi\rangle\langle\psi|)^{\otimes k}.
\end{equation}
In essence, a $t$-design provides a finite ensemble of pure states that is indistinguishable from Haar-random states for all tests involving up to $t$ copies of the state. This means that for any state polynomial of degree $t$ or less, its average over the design ensemble coincides with its Haar-average.

The projected ensemble framework aims to construct such designs by performing local measurements on a single many-body state which is evolved under quantum chaotic dynamics. Consider a quantum state $|\psi\rangle$ on $N = N_A + N_B$ qudits, partitioned into subsystems $A$ and $B$. A projective measurement on $B$ in a basis $\{|b\rangle\}$ yields post-measurement states on $A$, forming an ensemble $\mathcal{E}(|\psi\rangle, \mathcal{B}) = \{p_b, |\phi(b)\rangle\}$.
This ensemble approximates a $t$-design when $|\psi\rangle$ is sufficiently complex, such as after chaotic dynamics. The quality of approximation is quantified by the trace norm distance
\begin{equation}\label{eq:Delta_t}
\Delta^{(t)}_{\mathcal{E}}\equiv \left\|\sum_{|b\rangle\in\mathcal{B}}\dfrac{\left(\langle b|\psi\rangle\langle\psi |b\rangle\right)^{\otimes t}}{\left(\langle\psi |b\rangle\langle b|\psi\rangle\right)^{t-1}}-\int\limits_{|\phi\rangle\in\mathcal{E}^{A}_{\text{Haar}}}d\phi \left(|\phi\rangle\langle\phi|\right)^{\otimes t}\right\|_1\leq \varepsilon.
\end{equation}

The Haar average in Eq.~\eqref{eq:Delta_t} can be expressed as
\begin{equation}
\int d\phi\, (|\phi\rangle\langle\phi|)^{\otimes t} = \bm{\Pi}^{(t)}_A,
\end{equation}
where $\bm{\Pi}^{(t)}_A=\Pi^{(t)}_{A}/\mathrm{Tr}(\Pi^{(t)}_{A})$ and $\Pi^{(t)}=\sum_{j=1}^{t!}\pi_j/ t!$ is the projector onto the permutation symmetric subspace of $\mathcal{H}^{\otimes t}$, and $\mathcal{D}_A = d^{N_A}(d^{N_A}+1)\cdots(d^{N_A}+t-1)$. For Haar-random generator states, $\Delta^{(t)}_{\mathcal{E}}$ decays exponentially with $N_B$, independent of the measurement basis. However, when $|\psi\rangle$ respects global symmetries, the appropriate choice of the measurement basis becomes critical. Another important factor that influences whether state designs emerge is the restriction imposed by a resource theory. If the generator state belongs to the free-state set of a given QRT, then the resulting projected ensembles may fail to approximate higher-order state designs. This is one of the aspects we have investigated in this work.

\section{Construction of random unitaries with symmetry}
\label{polar}
In this section, we outline the procedure for uniformly sampling random unitaries that commute with the symmetry operator $\sigma^{\otimes N}_z$ \cite{varikuti2024unraveling}. A standard method for generating Haar-random unitaries from $\mathcal{U}(2^N) $ is to sample a complex Gaussian matrix and then perform the QR decomposition, which factors the input matrix into a unitary component and an upper-triangular component \cite{mezzadri2006generate}. After an appropriate normalization, the resulting unitary matrix is distributed according to the Haar measure over $ \mathcal{U}(2^N) $. However, the QR decomposition fails to preserve symmetries of the input matrix, such as $Z_2$ or translation symmetries. Hence, this method cannot be extended to sample unitaries abiding by symmetries. To circumvent this, we utilize the polar decomposition to generate these unitaries. In the following, we first show that the polar decomposition can be used to generate random unitaries that reproduce the moments of Haar unitaries.

Before proceeding, we first show that if the input matrix has independent complex Gaussian entries with zero mean and unit variance (also known as a Ginibre matrix), then the unitary obtained from its polar decomposition behaves like a Haar-random unitary. This is established by computing the moments of the resulting unitaries (denoted with $\mathcal{E}_{\text{polar}}$) and showing that they coincide with the corresponding Haar moments. Given an arbitrary complex Gaussian matrix $Z$, its right polar decomposition yields $ Z = UP $, where $ P = \sqrt{Z^\dagger Z} $ is a positive semi-definite matrix. When $ Z $ is of full rank, $ U $ can be uniquely defined as $ U = Z(Z^\dagger Z)^{-1/2} $. {Note that the set of rank-deficient matrices has measure zero with respect to the Ginibre ensemble. Indeed, rank deficiency is equivalent to $\det(Z)=0$, which defines the zero set of a nontrivial polynomial in the matrix entries. Such sets are known to have Lebesgue measure zero \cite{Mityagin2020ZeroSet}. Since the Ginibre ensemble is absolutely continuous with respect to Lebesgue measure, it follows that rank-deficient matrices occur with probability zero. Consequently, a Ginibre matrix is almost surely of full rank.}

To see how the Haar moments emerge from the polar decomposition, consider an arbitrary operator $ A $ acting on $ t $ replicas of a Hilbert space $ \mathcal{H}^d $. The average action of the ensemble on this operator is given by
\begin{eqnarray}
\langle U^{\dagger \otimes t} A U^{\otimes t} \rangle_{U\in \mathcal{E}_{\text{polar}}} &=& \int d\mu(Z)\, \left( (Z^\dagger Z)^{-1/2} Z^\dagger \right)^{\otimes t} A \left( Z (Z^\dagger Z)^{-1/2} \right)^{\otimes t}\nonumber\\
&=& \int d\mu(Z)\, \left( (Z^\dagger Z)^{-1/2} Z^\dagger \right)^{\otimes t} \left( \int d\mu(V)\, V^{\dagger \otimes t} A V^{\otimes t} \right) \left( Z (Z^\dagger Z)^{-1/2} \right)^{\otimes t}\nonumber\\
&=& \int d\mu(V)\, V^{\dagger \otimes t} A V^{\otimes t},
\end{eqnarray}
{where $d\mu(Z)$ denotes the probability measure associated with the Ginibre ensemble given by $d\mu(Z)\propto e^{-\mathrm{Tr}(Z^\dagger Z)}dZ$, which is invariant under left and right multiplication by unitary matrices $Z \mapsto UZV$}. The second {equality} follows from the unitary invariance of the Ginibre ensemble. The third equality implies that $\mathcal{E}_{\text{polar}}$ exactly matches those of the Haar measure.

We now show that if the initial operator abides by a symmetry, the resulting unitary from polar decomposition preserves that symmetry. For illustrative purposes, we consider the $Z_2$-symmetric case. Given a complex Ginibre matrix $Z$,
\begin{eqnarray}
Z' = Z+\sigma^{\otimes N}_{z} Z \sigma^{\otimes N}_{z}
\end{eqnarray}
is $Z_2$ symmetric. Provided $ Z' $ has full rank, one can check that the corresponding positive semi-definite part $ P' = \sqrt{Z'^\dagger Z'} $ also commutes with $\sigma^{\otimes N}_z$, and so does the unitary part $ U' = Z'(Z'^\dagger Z')^{-1/2} $. Furthermore, the distribution of $ Z' $ is invariant under conjugation by any unitary that commutes with the symmetry operator. Hence, the unitaries obtained via polar decomposition from $Z'$ form an ensemble that samples uniformly from the subgroup of $Z_2$-symmetric unitaries. We find this construction highly convenient for our numerical computations, and it may also be of interest for use beyond the scope of this work. {In the following, we outline the numerical algorithm used to generate random $Z_2$-symmetric unitaries. This construction can be readily extended to other symmetries, such as translation symmetry.

\begin{algorithm}[H]
\caption{Sampling a random $Z_2$-symmetric unitary via polar decomposition}
\begin{algorithmic}[1]
\State \textbf{Input:} Number of qubits $N$
\State \textbf{Output:} Random unitary $U$ such that $[U,\sigma_z^{\otimes N}] = 0$

\State $d \gets 2^N$

\State \textbf{Step-1: Generate Ginibre matrix:}
\State Sample $Z \in \mathbb{C}^{d \times d}$ with i.i.d.\ entries
$Z_{ij} \sim \mathcal{N}(0,1) + i\,\mathcal{N}(0,1)$

\State \textbf{Step-2: Impose $Z_2$ symmetry:}
\State $S_z \gets \sigma_z^{\otimes N}$
\State $Z_{\text{sym}} \gets Z + S_z Z S_z$

\State \textbf{Step-3: Check full rank:}
\If{$\det(Z_{\text{sym}}) \approx 0$}
    \State Resample $Z$ and return to Step 3
\EndIf

\State \textbf{Step-4: Compute positive part:}
\State $P \gets \sqrt{Z_{\text{sym}}^\dagger Z_{\text{sym}}}$

\State \textbf{Step-5: Polar decomposition:}
\State $U \gets Z_{\text{sym}} P^{-1}$

\State \Return $U$
\end{algorithmic}
\end{algorithm}
}

\section{Haar integration over the unitary subgroup $\mathcal{U}_{\mathbb{Z}_{2}}(2^N)$}
\label{appb}
In this section, we provide details concerning the Haar integration over the compact group $\mathcal{U}_{\mathbb{Z}_2}(2^N)$, comprising all the unitaries that either commute or anticommute with the symmetry operator $\sigma^{\otimes N}_z$. Extension to other discrete symmetries is straightforward. Here, we evaluate the first two moments of $\mathcal{U}_{\mathbb{Z}_2}(2^N)$. Before proceeding, it is useful to recall the notation for the subspace projectors
\begin{equation}
\mathbf{Z}_{k}=\dfrac{\mathbb{I}+(-1)^{k}\sigma^{\otimes N}_z}{2}, \;\text{ where }\; k\in \{0, 1\} \;\text{ and }\; \mathbf{Z}^{2}_{k}=\mathbf{Z}_{k}.
\end{equation}
In the following, we evaluate the first two moments associated with the Haar integration over $\mathcal{U}_{\mathbb{Z}_{2}}(2^N)$.

\subsection{First moment}
We are interested in evaluating the integral
\begin{equation}
\Theta^{(1)} \equiv \int_{F \in \mathcal{U}_{\mathbb{Z}_{2}}(2^N)} d\mu(F)\, \big(F^{\dagger} A F\big),
\end{equation}
for an arbitrary operator $A$ {acting non-trivially on} the Hilbert space $\mathcal{H}^{2^N}$. To solve this integral, we first identify the first-order commutant of $\mathcal{U}_{\mathbb{Z}_2}(2^N)$. By construction, $\mathcal{U}_{\mathbb{Z}_2}(2^N)$ can be written as the disjoint union $\mathcal{U}_{\mathbb{Z}_2}(2^N) = \mathcal{U}_{C} \cup \mathcal{U}_{AC}$, where the subset $\mathcal{U}_{C}$ consists of unitaries that commute with the symmetry operator $\sigma_z^{\otimes N}$, and the subset $\mathcal{U}_{AC}$ consists of unitaries that anticommute with it. Since one subset commutes and the other anticommutes with the symmetry operator, there exists no nontrivial operator that simultaneously commutes with both classes of unitaries. Consequently, the only element in the first-order commutant is the identity operator $\{ \mathbb{I}_{2^N} \}$.

Therefore, when performing the Haar integration over $\mathcal{U}_{\mathbb{Z}_{2}}(2^N)$, the integral acts as a twirl that projects $A$ onto the identity component: $\Theta^{(1)} = \alpha\, \mathbb{I}_{2^N}$, where the constant $\alpha$ is determined by the trace condition and is given by $\alpha=\text{Tr}(A)/\text{Tr}(\mathbb{I})=\text{Tr}(A)/2^N$. It then follows that
\begin{equation}
\Theta^{(1)} \equiv \int_{F \in \mathcal{U}_{\mathbb{Z}_{2}}(2^N)} d\mu(F)\, \big(F^{\dagger} A F\big) = \dfrac{\operatorname{Tr}(A)}{2^N}\, \mathbb{I}_{2^N}.
\end{equation}

\subsection{Second moment}
We are now interested in evaluating the second moment of $\mathcal{U}_{\mathbb{Z}_2}(2^N)$, given by
\begin{eqnarray}
  \Theta^{(2)}\equiv \int_{F\in\mathcal{U}_{\mathbb{Z}_{2}}(2^N)}d\mu(F)\left( F^{\dagger\otimes 2}AF^{\otimes 2} \right) ,
\end{eqnarray}
where $A$ has the support over the replica Hilbert space $\mathcal{H}^{2^N}\otimes \mathcal{H}^{2^N}$. To evaluate the above integral, we first identify its second-order commutant. While for every $u\in \mathcal{U}_{C}$, $u^{\otimes 2}$ commutes with operators from the cartesian product of the sets $\mathcal{S}_{C}\equiv \{ \mathbf{Z}^{\otimes 2}_{0},  \mathbf{Z}^{\otimes 2}_{1}, \mathbf{Z}^{(2)\perp}\}\times \{ {\Pi}^{+}, {\Pi}^{-}\}$, where ${\Pi}^{\pm}$ denote the projectors onto permutation symmetric ($+$) and anti-symmetric ($-$) subspaces, and $\mathbf{Z}^{(2)\perp}=\mathbb{I}_{2^N}-\mathbf{Z}^{\otimes 2}_{0}-\mathbf{Z}^{\otimes 2}_{1}$ denotes the complement of the $Z_2$-symmetric subspace in the two-replica Hilbert space. On the other hand, for every $w\in \mathcal{U}_{AC}$, one can see that $w^{\otimes 2}$ commutes with the cartesian product set $\mathcal{S}_{AC}\equiv \{ \mathbf{Z}^{\otimes 2}_{0}+\mathbf{Z}^{\otimes 2}_{1}, \mathbf{Z}^{(2)\perp}\}\times \{ {\Pi}^{+}, {\Pi}^{-}\}$. Therefore, the only common set of operators with which all $v^{\otimes 2}$ for $v\in\mathcal{U}_{\mathbb{Z}_2}$ commute is $\mathcal{S}_{C}\bigcap \mathcal{S}_{AC}\equiv \mathcal{S}_{AC}$.
It then follows that
\begin{eqnarray}\label{moment2}
\int_{F\in U_{\text{Z}}(2^N)}d\mu(F)\left(F^{\otimes 2} A F^{\dagger\otimes 2}\right)=\alpha\left(\mathbf{Z}^{\otimes 2}_{0}{\Pi^{+}}+ \mathbf{Z}^{\otimes 2}_{1}{\Pi}^{+}\right)+ \alpha_{\perp}\mathbf{Z}^{(2)\perp}{\Pi^{+}} + \beta\left(\mathbf{Z}^{\otimes 2}_{0}{\Pi}^{-}+ \mathbf{Z}^{\otimes 2}_{1}{\Pi}^{-}\right)+ \beta_{\perp} \mathbf{Z}^{(2)\perp}{\Pi}^{-}.
\end{eqnarray}
The constants $\alpha, \alpha_{\perp}, \beta$ and $\beta_{\perp}$ ensure proper normalization of the right-hand side.

\section{Main results concerning the $Z_2$-asymmetry generating power ($Z_2$-AGP)}
\label{appc}

In this section, we present rigorous proofs of all results concerning the \(Z_{2}\)-AGP discussed in the main text, including the twirling identity (Eq.~(5)), the AGP analogue of Theorem~1, and the thermalization behavior of the \(Z_{2}\)-AGP. For completeness, we also provide a few supplementary results that were omitted from the main text.

\subsection{Derivation of the twirling identity in Eq.~(5) (main text)}
\label{appc-1}

Here, we provide a complete derivation of Eq.~(5) using the results from the previous section concerning Haar integrals over the unitary subgroup \(\mathcal{U}_{Z_2}(2^N)\). {We start from Eq.~(\ref{moment2}) of the previous section, which reads}
\begin{eqnarray}\label{moment2-1}
\int_{F\in U_{\mathbb{Z}_{2}}(2^N)}d\mu(F)\left(F^{\otimes 2} A F^{\dagger\otimes 2}\right)=\alpha\left(\mathbf{Z}^{\otimes 2}_{0}{\Pi^{+}}+ \mathbf{Z}^{\otimes 2}_{1}{\Pi}^{+}\right)+ \alpha_{\perp}\mathbf{Z}^{(2)\perp}{\Pi^{+}} + \beta\left(\mathbf{Z}^{\otimes 2}_{0}{\Pi}^{-}+ \mathbf{Z}^{\otimes 2}_{1}{\Pi}^{-}\right)+ \beta_{\perp} \mathbf{Z}^{(2)\perp}{\Pi}^{-}.
\end{eqnarray}
Here, we take $A = U^{\otimes 2}\rho_f U^{\dagger \otimes 2}$, where {$U$} is assumed to have finite $\mathcal{A}_p$, and $\rho_f$ denotes the second moment of the ensemble of free states, which we evaluate to be the following:
\begin{equation}
\rho_f =\langle \left( |\psi_{\mathrm{free}}\rangle\langle\psi_{\mathrm{free}} | \right)^{\otimes 2}  \rangle_{|\psi\rangle_{\mathrm{free}}}= \dfrac{1}{2}\left[  \dfrac{\mathbf{Z}^{\otimes 2}_{0}\Pi^{(2)} }{\text{Tr}\left( \mathbf{Z}^{\otimes 2}_{0}\Pi^{(2)} \right)} + \dfrac{\mathbf{Z}^{\otimes 2}_{1}\Pi^{(2)} }{\text{Tr}\left( \mathbf{Z}^{\otimes 2}_{1}\Pi^{(2)} \right)}   \right]
\end{equation}
{Note that $A = U^{\otimes 2} \rho_{f} U^{\dagger \otimes 2}$ acts nontrivially only on the permutation-symmetric subspace. This follows from the fact that $\Pi^{-} A \Pi^{-} = 0$, i.e., the antisymmetric subspace lies in the kernel of $A$. As a result, the coefficients $\beta$ and $\beta_{\perp}$ vanish.}
Consequently, Eq.~(\ref{moment2}) becomes
\begin{eqnarray}\label{FAF}
\int_{F\in U_{\text{Z}}(2^N)}d\mu(F)\left(F^{\otimes 2} A F^{\dagger\otimes 2}\right)=\alpha\left(\mathbf{Z}^{\otimes 2}_{0}+ \mathbf{Z}^{\otimes 2}_{1}\right) {\Pi}^{(2)}+ \alpha_{\perp}\mathbf{Z}^{(2)\perp}{\Pi^{(2)}},
\end{eqnarray}
In the above equation, we have denoted $\Pi^{+}$ with $\Pi^{(2)}$, consistent with the notation choice throughout this paper. The coefficients $\alpha$ and $\alpha_{\perp}$ can be evaluated as
\begin{eqnarray}
\alpha= \dfrac{\text{Tr}\left[\left(\mathbf{Z}^{\otimes 2}_{0}+\mathbf{Z}^{\otimes 2}_{1}\right){\Pi}^{(2)} A \right]}{\text{Tr}\left[ \left(\mathbf{Z}^{\otimes 2}_{0} + \mathbf{Z}^{\otimes 2}_{1}\right){\Pi}^{(2)}\right]}\;\text{ and }\; \alpha_\perp= \dfrac{\text{Tr}\left(\mathbf{Z}^{(2)\perp}{\Pi^{(2)}} A \right)}{\text{Tr}\left( \mathbf{Z}^{(2)\perp}{\Pi^{(2)}}\right)}.
\end{eqnarray}
By noting that $\text{Tr}\left( \mathbf{Z}^{\otimes 2}_{0}\Pi^{(2)} \right)=\text{Tr}\left( \mathbf{Z}^{\otimes 2}_{1}\Pi^{(2)} \right)$, we rewrite Eq.~(\ref{FAF}) as
\begin{eqnarray}\label{secmoment}
\int_{F\in U_{\text{Z}}(2^N)}d\mu(F)F^{\otimes 2} A F^{\dagger\otimes 2}&=&\text{Tr}\left[\left( \mathbf{Z}^{\otimes 2}_{0} + \mathbf{Z}^{\otimes 2}_{1} \right){\Pi}^{(2)}  A \right] \left\{\dfrac{1}{2}\left[  \dfrac{\mathbf{Z}^{\otimes 2}_{0}\Pi^{(2)} }{\text{Tr}\left( \mathbf{Z}^{\otimes 2}_{0}\Pi^{(2)} \right)} + \dfrac{\mathbf{Z}^{\otimes 2}_{1}\Pi^{(2)} }{\text{Tr}\left( \mathbf{Z}^{\otimes 2}_{1}\Pi^{(2)} \right)}   \right]\right\}+ \dfrac{\text{Tr}\left[ \mathbf{Z}^{(2)\perp}{\Pi}^{(2)}A \right]}{\text{Tr}\left[ \mathbf{Z}^{(2)\perp}{\Pi}^{(2)} \right]}\mathbf{Z}^{(2)\perp}{\Pi}^{(2)} \nonumber\\
&=&\text{Tr}\left[\left( \mathbf{Z}^{\otimes 2}_{0} + \mathbf{Z}^{\otimes 2}_{1} \right){\Pi}^{(2)}  A \right] \rho_{\text{f}}+ \dfrac{\text{Tr}\left[ \mathbf{Z}^{(2)\perp}{\Pi}^{(2)}A \right]}{\text{Tr}\left[ \mathbf{Z}^{(2)\perp}{\Pi}^{(2)} \right]}\mathbf{Z}^{(2)\perp}{\Pi}^{(2)}.
\end{eqnarray}
With a few simplifications, it follows that
\begin{eqnarray}\label{eq4}
\left\langle F^{\otimes 2} U^{\otimes 2}\rho_{\text{f}}U^{\dagger\otimes 2} F^{\dagger\otimes 2} \right\rangle_{F} &=& \text{Tr}\left[\left( \mathbf{Z}^{\otimes 2}_{0} + \mathbf{Z}^{\otimes 2}_{1} \right) U^{\otimes 2}\rho_{\text{f}}U^{\dagger\otimes 2} \right] \rho_{\text{f}}+ \dfrac{\text{Tr}\left[ \mathbf{Z}^{(2)\perp}U^{\otimes 2}\rho_{\text{f}}U^{\dagger \otimes 2} \right]}{\text{Tr}\left[ \mathbf{Z}^{(2)\perp}{\Pi} \right]}\mathbf{Z}^{(2)\perp}{\Pi}^{(2)}\nonumber\\
&=&  \left[ 1-\mathcal{A}_{p}(U) \right]\rho_f + \dfrac{\mathcal{A}_{p}(U)}{\overline{\mathcal{A}}} \dfrac{\Pi^{(2)}}{\mathrm{Tr}\left( \Pi^{(2)} \right)} - \dfrac{\mathcal{A}_{p}(U)}{\overline{\mathcal{A}}}\left( 1-\overline{\mathcal{A}} \right)\rho_f  \nonumber\\
&=&\left[ 1- \dfrac{\mathcal{A}_{p}(U)}{\overline{\mathcal{A}}} \right]\; \rho_{\text{f}}\;  + \; \dfrac{\mathcal{A}_{p}(U)}{\overline{\mathcal{A}}}\;\dfrac{\Pi^{(2)}}{\mathrm{Tr}\left( \Pi^{(2)} \right)},
\end{eqnarray}
 proving the  Eq.~(5) of the main text. In the second equality of the above equation, we have substituted $\mathcal{A}_p(U)=1-\text{Tr}\left[\left( \mathbf{Z}^{\otimes 2}_{0} + \mathbf{Z}^{\otimes 2}_{1} \right) U^{\otimes 2}\rho_{\text{f}}U^{\dagger\otimes 2} \right]$.

\subsection{Theorem~1 (main text): $Z_2$-AGP version}

Having established the above result, we are now in a position to prove Theorem~1 for the case of the QRT of $Z_2$-asymmetry. Recall the protocol given in the main text:\\
(i) prepare a random bipartite free state $|\psi\rangle\in\mathcal{H}_{AB}$; \\
(ii) apply the unitary $U$, whose resource content is to be estimated; \\
(iii) apply a random free operation $F$;\\
(iv) measure a subset of the final state (say $B$) in a free basis and record the probabilities $\{p_b\}$. \\

The resulting $p_b$ can be written as
\begin{eqnarray}
 p_b = \langle\psi |U^{\dagger}F^{\dagger} |b\rangle\langle b|FU|\psi\rangle .
\end{eqnarray}
Upon squaring and performing averages over the entire ensemble of free states and the free unitaries, we get
\begin{eqnarray}
 \langle p^{2}_b\rangle_{F, |\psi\rangle}  &=& \left\langle\text{Tr}\left[ \left( |b\rangle\langle b| \right)^{\otimes 2} \left( FU|\psi\rangle\langle\psi |U^{\dagger}F^{\dagger} \right)^{\otimes 2}\right]\right\rangle_{F, |\psi\rangle}\nonumber\\
 &=& \left\langle\text{Tr}\left[ \left( |b\rangle\langle b| \right)^{\otimes 2} \left( FU \rho_{\text{f}} U^{\dagger}F^{\dagger} \right)^{\otimes 2}\right]\right\rangle_{F}\nonumber\\
 &=& \left[ 1- \dfrac{\mathcal{A}_{p}(U)}{\overline{\mathcal{A}}} \right]\; \text{Tr}\left(\langle b^{\otimes 2} |\rho_{\text{f}}|b^{\otimes 2}\rangle\right)\;  + \; \dfrac{\mathcal{A}_{p}(U)}{\overline{\mathcal{A}}}\;\dfrac{\text{Tr}\left(\langle b^{\otimes 2}|\Pi^{(2)}|b^{\otimes 2}\rangle\right)}{\mathrm{Tr}\left( \Pi^{(2)} \right)}, \label{eq:pb2Fpsi}
\end{eqnarray}
where in the second equality, we replaced the averaging over the free states with $\rho_{\text{f}}$ and in the third equality, we used the relation in Eq.~(\ref{eq4}). For simplicity, we write $k_1=\dfrac{\text{Tr}\left(\langle b^{\otimes 2}|\Pi^{(2)}|b^{\otimes 2}\rangle\right)}{{\mathrm{Tr}\left( \Pi^{(2)} \right)}}$ and $k_2=\text{Tr}\left(\langle b^{\otimes 2} |\rho_{\text{f}}|b^{\otimes 2}\rangle\right)$. These two constants depend only on the total system dimension $2^N$ and the subsystem dimension $2^{N_A}$. When $N_A$ is finite, these constants can be evaluated to read
\begin{eqnarray}\label{agpk1}
k_1=\dfrac{\text{Tr}\left( \langle b^{\otimes 2}|\Pi^{(2)}|b^{\otimes 2}\rangle \right)}{\mathrm{Tr}\left( \Pi^{(2)} \right)} = \dfrac{2^{N_A}(2^{N_A}+1)}{2^N(2^N+1)}
\end{eqnarray}
and
\begin{eqnarray}\label{agpk2}
k_2=\text{Tr}\left( \langle b^{\otimes 2}|\rho_{\text{f}}|b^{\otimes 2}\rangle \right)=\dfrac{1}{2} \sum_{k\in\{0, 1\}} \dfrac{\text{Tr} \left(\mathbf{Z}^{\otimes 2}_{k+\text{sgn}(b), A }\Pi^{(2)}_{A}\right)}{\text{Tr}\left( \mathbf{Z}^{\otimes 2}_{k}\Pi^{(2)}_{AB} \right)}= \dfrac{2^{N_A}(2^{N_A}+2)}{2^N(2^N+2)} .
\end{eqnarray}
Following this, one can rearrange Eq.~\eqref{eq:pb2Fpsi} to write the $Z_2$-AGP of an arbitrary evolution in terms of the success probability for the outcome corresponding to $|b\rangle$ as
\begin{eqnarray}\label{estAGP}
 \dfrac{\mathcal{A}_{p}(U)}{\overline{A}} = \dfrac{k_2-\langle p^2_b\rangle_{F, |\psi\rangle}}{k_2-k_1}  ,
\end{eqnarray}
which corresponds to $Z_2$-AGP analogue of Eq.~(1) in the main text. The above expression provides an unbiased estimator for the $Z_2$-AGP.

If one has access to the probabilities corresponding to all the basis vectors, then the unbiased estimator of the $Z_2$-AGP can be written as
\begin{eqnarray}\label{z2agpest}
\mathcal{A}_{\mathrm{est}} = \dfrac{k_2-\tilde{P}^{(2)}}{k_2-k_1}, \;\text{ where }\; \tilde{P}^{(2)} =\dfrac{1}{2^{N_B}}\sum_{b=0}^{2^{N_B}-1} p^{2}_b,
\end{eqnarray}
which can further reduce the statistical error associated with the finite sampling of free unitaries and states.

\begin{figure}
\begin{center}
\includegraphics[scale=0.6]{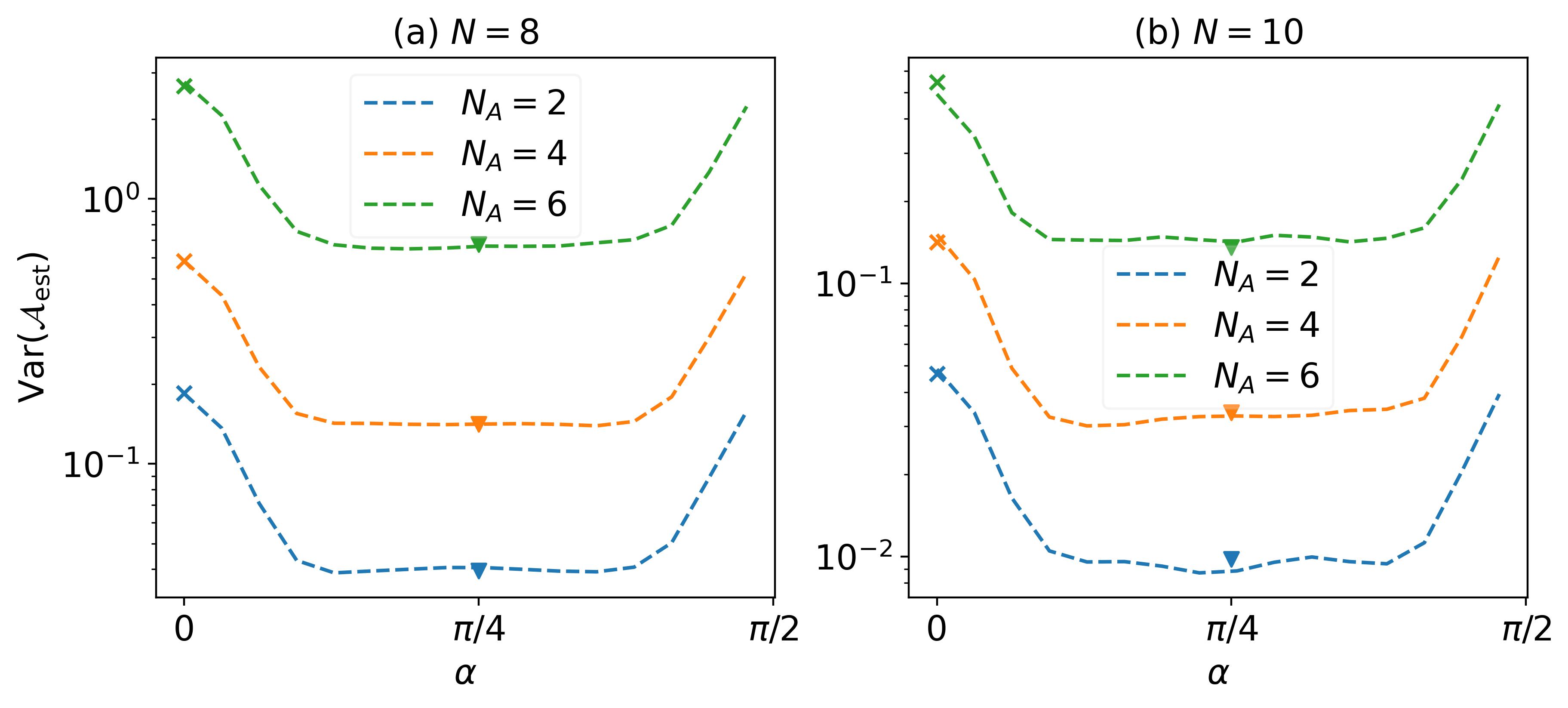}
\caption{\label{fig:AGP-prot-fluct} The figure illustrates the variance of the estimator $\mathcal{A}_{\mathrm{est}}$ (see Eq.~(\ref{z2agpest})) for the system dynamics considered in the main text ($U = u^{\otimes N}, \;u=e^{-i\alpha\sigma_x }$) for two different $N$, namely, $N=8$ in (a) and $N=10$ in (b), each for three different unmeasured subsystem sizes, namely, $N_A=2, 4$ and $6$. The figure shows that, for fixed $N$, the estimator’s variance grows exponentially with $N_A$. When $\alpha = 0$, $U$ remains a free unitary, and the variance for all cases agrees with the analytical expression in Eq.~(\ref{varasymfree}), as indicated by the points marked with `$\times$' in both panels. Near $\alpha = \pi/4$, the unitary approaches a regime where $\mathcal{A}_{p}$ is close to its Haar value, and the variance correspondingly matches the prediction of Eq.~(\ref{varasymhaar}), which is shown by the points marked with `$\blacktriangledown$'. While the results in panel (a) are obtained by averaging over $\sim 10^4$ samples, those in panel (b) are computed using a sample size of $\sim 10^3$.
}
\end{center}
\end{figure}

{
\subsubsection{Fluctuations of the estimator over free states and free unitaries}
Here, we analyze the sampling complexity for estimating the $Z_2$-AGP using the unbiased estimator given in Eq.~(\ref{z2agpest}). The variance can be evaluated as
\begin{eqnarray}\label{varstartasy}
\mathrm{Var}(\mathcal{A}_{\mathrm{est}}) &=& \dfrac{1}{(k_2-k_1)^2}\left[ \left\langle\left( \tilde{p}^{(2)}\right)^2 \right\rangle - \left\langle \tilde{p}^{(2)} \right\rangle^2  \right] \nonumber\\
&=& \dfrac{1}{(k_2-k_1)^2}\left[ \left\langle \dfrac{1}{2^{2N_B}} \sum_{b, b' \in \{0, 1\}^{N_{B}}} p^2_{b}p^2_{b'} \right\rangle - \left\langle \tilde{p}^{(2)} \right\rangle^2  \right]\,.
\end{eqnarray}
In the following, we consider the limiting cases, namely, (a) when $U$ is a free unitary and (b) when $U$ is a fully Haar random unitary, and compute the fluctuations analytically.

\noindent\textit{\textbf{(a) $U$ is a free unitary:}}
When $U$ is taken as a free unitary, it follows from Eq.~(\ref{estAGP}) that $\langle \tilde{p}^2\rangle = k_2$. On the other hand, $\left\langle \left(\tilde{p}^{(2)}\right)^{2}\right\rangle_{F, |\psi\rangle}$ can be computed by noticing the following:
\begin{eqnarray}
\left\langle p^{2}_{b}p^{2}_{b'} \right\rangle =
\begin{cases}
\displaystyle
\prod_{s=0}^{3}\dfrac{2^{N_{A}}+2s}{2^{N}+2s}
& \text{if } \; b = b' \\\\
\displaystyle
k^2_{2}\;\dfrac{2^N(2^N+2)}{(2^N+4)(2^N+6)}
& \text{otherwise}.
\end{cases}
\end{eqnarray}
Using these relations in Eq.~(\ref{varstartasy}), the variance is
\begin{eqnarray}\label{varasymfree}
\mathrm{Var}(\mathcal{A}_{\mathrm{est}}) &=& \dfrac{1}{(k_2-k_1)^2} \left[ \dfrac{1}{2^{N_{B}}} \left( \prod_{s=0}^{3}
\frac{2^{N_{A}} + 2s}{2^{N} + 2s} \right) + \dfrac{2^{N_B}-1}{2^{N_{B}}}  \left(
\dfrac{2^N (2^N+2)}{(2^N+4)(2^N+6)} \right) k^2_{2} - k^2_{2}  \right] \nonumber\\
&\approx & \dfrac{8}{2^{N_{B}}-1}\,,
\end{eqnarray}
where in the second line we approximated the variance for large $N$, while keeping $N_B$ finite and fixed.
To estimate the $Z_2$-AGP with an additive accuracy of $\varepsilon$, the number of free unitaries scales roughly as $\sim 8/[\varepsilon^2\cdot (2^{N_{B}}-1)]$. \\

\noindent\textit{(b) \textbf{$U$ is a Haar unitary:}} In this case, it is straightforward to see that
\begin{eqnarray}
\langle p^{2}_{b}p^{2}_{b'}\rangle_{U\in U(2^N)} =
\begin{cases}
\displaystyle \prod_{s=0}^{3}
\frac{2^{N_{A}} + s}{2^{N} + s},
& \text{if } \; b = b' \\[6pt]
\displaystyle k_1^2\,
\frac{2^N (2^N+1)}{(2^N+2)(2^N+3)},
& \text{otherwise}.
\end{cases}
\end{eqnarray}
Then, the variance of the estimator becomes
\begin{eqnarray}\label{varasymhaar}
\mathrm{Var}(\mathcal{A}_{\mathrm{est}}) &=& \dfrac{1}{(k_2-k_1)^2} \left[ \dfrac{1}{2^{N_{B}}} \left( \prod_{s=0}^{3}
\frac{2^{N_{A}} + s}{2^{N} + s} \right) + \dfrac{2^{N_B}-1}{2^{N_{B}}}  \left(
\dfrac{2^N (2^N+1)}{(2^N+2)(2^N+3)} \right) k^2_{1} - k^2_{1}  \right]\nonumber\\
&\approx & \dfrac{2}{2^{N_{B}}-1} \;\; \mathrm{ for }\;\; N_{A}\gg N_{B}.
\end{eqnarray}
In both cases, the estimator variance over free states and unitaries decreases exponentially with subsystem size. As our numerics in Fig.~\ref{fig:AGP-prot-fluct} corroborate, this exponential decrease also holds in scenarios interpolating between the two limiting cases.

}

\subsection{AGP counterpart of Eq.~(2) (main text)}

Here, we evaluate the average asymmetry of the projected states generated within the protocol considered in this work, retaining terms up to the leading order. The projected states asymmetry can be quantified by $\mathcal{A}(|\phi_b\rangle)=\langle\phi^{\otimes 2}_b| \mathbf{Z}^{(2)\perp}_{A} |\phi^{\otimes 2}_b\rangle$, where $\mathbf{Z}^{(2)\perp}=\mathbb{I}_{2^{2N_A}}-\mathbf{Z}^{\otimes 2}_{0, A}-\mathbf{Z}^{\otimes 2}_{1, A}$ denotes the projector onto the complement of the $Z_2$-symmetric subspace supported over $N_A$-qubits and $|b\rangle$ is a computational basis vector in $\mathcal{H}_{B}$.
We now compute the average asymmetry in the projected states as follows:
\begin{eqnarray}
\left\langle \mathcal{A}_{\mathrm{PS}}\right\rangle_{|\psi\rangle_{\text{free}}, F}=\left\langle \text{Tr}\left[ \mathbf{Z}^{(2)\perp}_{A} \dfrac{\left(\langle b| FU|\psi\rangle\langle\psi |U^{\dagger}F^{\dagger} |b\rangle^{\otimes 2}\right)}{ p^2_b } \right]\right\rangle_{F, |\psi\rangle_{\mathrm{free}}}.
\end{eqnarray}
In the above expression, the average is taken over both the free states and the free unitaries. For brevity, we henceforth omit the subscript “free’’ from $|\psi\rangle_{\mathrm{free}}$.
The presence of the normalizing factor in the above expression makes the Haar integration analytically challenging, except in special cases—for instance, when $U$ is a Haar-random unitary. To nevertheless get a handle on it, we can use a  general expansion of expectations involving ratios of functions,
\begin{equation}
\mathbb{E}\!\left[\frac{f}{g}\right]
\approx
\frac{\mu_f}{\mu_g}
\;-\; \frac{\mathrm{Cov}(f,g)}{\mu_g^{2}}
\;+\; \frac{\mu_f\,\mathrm{Var}(g)}{\mu_g^{3}} \,,
\end{equation}
where the leading-order term involves the independent averages over the functions. Following this, we obtain the leading-order correction to the average asymmetry of the projected states $\langle \mathcal{A}_{\mathrm{PS}}\rangle_{F, |\psi\rangle}$. First, we evaluate the numerator as
\begin{eqnarray}\label{num}
\left\langle\text{Tr}\left[ \mathbf{Z}^{(2)\perp}_{A} \left(\langle b| FU|\psi \rangle\langle\psi|U^{\dagger}F^{\dagger} |b\rangle^{\otimes 2}\right) \right]\right\rangle_{|\psi\rangle, F} &=& \left( 1-\dfrac{\mathcal{A}_p(U)}{\overline{\mathcal{A}}} \right) \text{Tr}\left( \mathbf{Z}^{(2)\perp}_{A} \langle b^{\otimes 2} |\rho_{\text{f}} |b^{\otimes 2}\rangle  \right) + \dfrac{\mathcal{A}_p(U)}{\overline{\mathcal{A}}}{\text{Tr}\left( \mathbf{Z}^{(2)\perp}_{A} \langle b^{\otimes 2}| \Pi^{(2)} |b^{\otimes 2}\rangle\right)} \nonumber\\
&=& \mathcal{A}_{p}(U) \dfrac{\overline{\mathcal{A}_{A}}}{\overline{\mathcal{A}}}\; k_1,
\end{eqnarray}
where in the first equality we made use of Eq.~(\ref{eq4}). The term in the denominator follows from the previous subsection,
\begin{eqnarray}\label{den}
\langle p^2_b\rangle_{F, |\psi\rangle} &=& \left( 1-\dfrac{\mathcal{A}_p(U)}{\overline{\mathcal{A}}} \right)\; k_2 + \dfrac{\mathcal{A}_p(U)}{\overline{\mathcal{A}}}\; k_1 .
\end{eqnarray}
{Note that the constants $k_1$ and $k_2$ in the above equations are the same as those obtained in the previous subsection [see Eqs.~(\ref{agpk1}) and ~(\ref{agpk2})].
Combining Eqs.~(\ref{num}) and ~(\ref{den}), the average asymmetry of the projected states can be written up to the zeroth-order correction as
\begin{eqnarray}
\left\langle \mathcal{A}_{\mathrm{PS}}\right\rangle_{|\psi\rangle, F} \approx \dfrac{\mathcal{A}_{p}(U) \dfrac{\overline{\mathcal{A}_{A}}}{\overline{\mathcal{A}}}\; k_1}{\left( 1-\dfrac{\mathcal{A}_p(U)}{\overline{\mathcal{A}}} \right)\; k_2 + \dfrac{\mathcal{A}_p(U)}{\overline{\mathcal{A}}}\; k_1 }.
\end{eqnarray}
This corresponds to the AGP version of Eq. (2) of the main text. This can be further simplified by considering the magnitude of $k_1/k2$:
\begin{equation}
\frac{k_1}{k_2}= \frac{(2^{N_A}+1)(2^N+2)}{(2^{N_A}+2)(2^N+1)}= 1 - \frac{1}{2^{N_A}} + \frac{1}{2^N}
+ O(2^{-2N_A}, 2^{-2N}).
\end{equation}
In the limit of large $N$ and $N_A$, and when $N_A$ is also comparable to $N$, $k_1/k_2\approx 1$. As a result, we obtain the following:
\begin{eqnarray}
\dfrac{\left\langle \mathcal{A}_{\mathrm{PS}}\right\rangle_{|\psi\rangle, F}}{\overline{\mathcal{A}_{A}}}\approx  \dfrac{\mathcal{A}_p(U)}{\overline{\mathcal{A}}}.
\end{eqnarray}
}
One can infer from the above equation that the leading-order correction to the average asymmetry of the projected states follows the same dynamics as the AGP of the overall unitary $U$, given by $\mathcal{A}_p(U)$.
In the main text, we use this result to probe deep thermalization of the $Z_2$-asymmetry under generic quantum circuits composed of free and resourceful operations. Corresponding results are shown in Fig.~2(b) of the main text.

\subsection{Thermalization of the $Z_2$-AGP (Theorem 2 of the main text)}

Here, we prove Theorem~2 for the $Z_2$-AGP (also see Eq.~(6) of the main text), which states that the $Z_2$-AGP of the circuit consisting of interlacing free and non-free operations converges exponentially to its Haar-averaged value. Before proceeding, we first prove the following important result concerning the AGP of a free unitary sandwiched between two arbitrary non-free operations.

\begin{lemma}\label{lemma} (\normalfont{Impact of free operations on $Z_2$-AGP})
Let $U$ and $V$ be two arbitrary {unitaries} in the QRT of $Z_2$-asymmetry over an $N$-qubit Hilbert space $\mathcal{H}^{2^{N}}$, with their corresponding AGPs given by $\mathcal{A}_{p}(U)$ and $\mathcal{A}_{p}(V)$, respectively. Let $F\in U_{Z_2}(2^N)$ be a random free operation drawn from $\mathcal{U}_{\mathbb{Z}_{2}}(2^N)$ according to its Haar measure. Then, the AGP of the free unitary sandwiched between $U$ and $V$, on average, is related to $\mathcal{A}_{p}(U)$ and $\mathcal{A}_{p}(V)$ as
\begin{equation}
\left\langle \mathcal{A}_{p}(VFU) \right\rangle_{F} = \mathcal{A}_{p}(U)+\mathcal{A}_{p}(V)-\dfrac{\mathcal{A}_{p}(U)\mathcal{A}_{p}(V)}{\overline{\mathcal{A}}},
\end{equation}
\end{lemma}
\begin{proof}
We are interested in evaluating $\left\langle\mathcal{A}_{p}(VFU)\right\rangle_{F}$, the average $Z_2$-AGP, when two arbitrary symmetry-breaking operations $U$ and $V$ are interspersed with a random free operation $F$ abiding by the symmetry. This can be expressed as the following integral over the free operations:
\begin{eqnarray}\label{AGPintm1}
\left\langle\mathcal{A}_{p}(VFU)\right\rangle_{F}=1-\int_{F}d\mu(F)\text{Tr}\left[ \left(\mathbf{Z}^{\otimes 2}_{0}+\mathbf{Z}^{\otimes 2}_{1}\right) \left(VFU\right)^{\otimes 2}\rho_{\text{f}}\left(VFU\right)^{\dagger \otimes 2} \right].
\end{eqnarray}
To solve the above integral, we make use of Eq.~(\ref{eq4}), given by
\begin{eqnarray}
\int_{F}d\mu(F)F^{\otimes 2} U^{\otimes 2}\rho_{\text{f}}U^{\dagger\otimes 2} F^{\dagger\otimes 2}=\left[ 1- \dfrac{\mathcal{A}_{p}(U)}{\overline{\mathcal{A}}} \right]\; \rho_{\text{f}}\;  + \; \dfrac{\mathcal{A}_{p}(U)}{\overline{\mathcal{A}}}\;\dfrac{\Pi^{(2)}}{\mathrm{Tr}\left( \Pi^{(2)} \right)},
\end{eqnarray}
Upon substituting the above equation in Eq.~(\ref{AGPintm1}), we get
\begin{eqnarray}
\left\langle\mathcal{A}_{p}(VFU)\right\rangle_{F} &=& 1- \left( 1-\dfrac{\mathcal{A}_{p}(U)}{\overline{\mathcal{A}}} \right) \text{Tr}\left[ \left( \mathbf{Z}^{\otimes 2}_{0}+\mathbf{Z}^{\otimes 2}_{1} \right)V^{\otimes 2}\rho_{\text{f}}V^{\dagger\otimes 2} \right]   - \dfrac{\mathcal{A}_{p}(U)}{\overline{\mathcal{A}}} \text{Tr}\left[ \left( \mathbf{Z}^{\otimes 2}_{0}+\mathbf{Z}^{\otimes 2}_{1} \right)V^{\otimes 2}\dfrac{\Pi^{(2)}}{\mathrm{Tr}\left( \Pi^{(2)} \right)}V^{\dagger\otimes 2} \right]\nonumber\\
&=&1- \left( 1-\dfrac{\mathcal{A}_{p}(U)}{\overline{\mathcal{A}}} \right) \left( 1-\mathcal{A}_{p}(V) \right) - \dfrac{\mathcal{A}_{p}(U)}{\overline{\mathcal{A}}} \left(1-\overline{\mathcal{A}}\right)
\end{eqnarray}
With a slight rearrangement of the terms, we get
\begin{equation}\label{asym_impact}
\left\langle\mathcal{A}_{p}(VFU)\right\rangle_{F}=\mathcal{A}_{p}(U)+\mathcal{A}_{p}(V)-\dfrac{\mathcal{A}_{p}(U)\mathcal{A}_{p}(V)}{\overline{\mathcal{A}}},
\end{equation}
hence proving the lemma.
\end{proof}

An implication of the above result is that, under repeated interspersions of free and non-free operations, the system gradually loses the memory of its initial symmetry, with this decay occurring exponentially fast in time. This can be verified by showing that the AGP of the unitary circuit relaxes exponentially with time to its Haar-averaged value.  We state this result in the following theorem.

\begin{theorem}\normalfont{(Thermalization of $Z_2$-AGP)}

Let $U^{(n)} = U F_{n-1} U \cdots F_1 U$, where each $F_j$ (for all $1 \leq j \leq n-1$) is an independently drawn random unitary from the group $\mathcal{U}_{\mathbb{Z}_2}(2^N)$ of $Z_2$-symmetry-preserving unitaries, sampled according to the Haar measure. Let $U$ be a fixed symmetry-breaking unitary with a finite AGP given by $\mathcal{A}_{p}(U)$. Then the evolution of AGP satisfies
\begin{equation}
\left\langle \mathcal{A}_{p}(U^{(n)}) \right\rangle_{\tilde{F}} = \overline{\mathcal{A}} \left[ 1 - \left( 1 - \frac{\mathcal{A}_{p}(U)}{\overline{\mathcal{A}}} \right)^n \right],
\end{equation}
where the expectation $\langle \cdot \rangle_{\tilde{F}}$ denotes averaging over all the independent random free operations $\{F_j\}$ at each time step.
\end{theorem}

\begin{proof}
Given an arbitrary non-free unitary $U$ with a finite AGP, Lemma~\ref{lemma} implies the following:
\begin{eqnarray}
\langle \mathcal{A}_{p}(VFU) \rangle_{F}&=&\mathcal{A}_{p}(U)+\mathcal{A}_{p}(V)-\dfrac{\mathcal{A}_{p}(U)\mathcal{A}_{p}(V)}{\overline{\mathcal{A}}}\nonumber\\
&=&\overline{\mathcal{A}}\left[ 1-\left(1-\dfrac{\mathcal{A}_{p}(U)}{\overline{\mathcal{A}}}\right)\left(1-\dfrac{\mathcal{A}_{p}(V)}{\overline{\mathcal{A}}}\right) \right].
\end{eqnarray}
We consider $V=U^{(t-1)}=UF_{t-1}\cdots F_2U$. Then, recursive integrations over the free operations lead to
\begin{equation}\label{expasym_supp}
\left\langle \mathcal{A}_{p}(U^{(n)}) \right\rangle_{\tilde{F}} = \overline{\mathcal{A}} \left[ 1 - \left( 1 - \frac{\mathcal{A}_{p}(U)}{\overline{\mathcal{A}}} \right)^n \right],
\end{equation}
proving the result.
\end{proof}

\vspace{1em}

This result appears in Eq.~(6) of the main text. We numerically verify the thermalization of the $Z_2$-AGP by considering a fixed unitary of the form $U = u^{\otimes N}$, where $u = \exp\{-i\pi\sigma_x/24\}$, for several values of $N$. The corresponding results are shown in Fig.~\ref{fig:AGP-therm} of this supplemental text. We observe that the numerical data agree with the analytical expression in Eq.~(\ref{expasym_supp}). In a more general scenario when $U^{(n)}=U_{n}F_{n-1}\cdots F_1 U_1$, where each $U_j$ for all $1\leq j\leq n$ may have different asymmetry generating powers, Eq.~(\ref{expasym_supp}) becomes
\begin{eqnarray}
 \langle \mathcal{A}_{p}(U^{(n)})\rangle_{\tilde{F}}   = \overline{\mathcal{A}}\left[ 1- \prod_{j=1}^{n}\left( 1-\dfrac{\mathcal{A}_{p}(U_j)}{\overline{\mathcal{A}}} \right) \right],
\end{eqnarray}
 which also converges to the Haar value exponentially with time. In the following, we examine the behavior of Eq.~(\ref{expasym_supp}) for two cases, where $\mathcal{A}_{p}(U)\leq \overline{\mathcal{A}}$, and $\mathcal{A}_{p}(U)> \overline{\mathcal{A}}$.

{
\textbf{Case-1} ($\mathcal{A}_{p}(U)\leq \overline{\mathcal{A}}$):
In this case, Eq.~(\ref{expasym_supp}) can be rewritten as follows:
\begin{eqnarray}
\left\langle \mathcal{A}_{p}(U^{(n)}) \right\rangle_{\tilde{F}}
&=& \overline{\mathcal{A}} \left[ 1 - \exp\left\{ n \ln\left( 1 - \frac{\mathcal{A}_{p}(U)}{\overline{\mathcal{A}}} \right) \right\} \right].
\end{eqnarray}
It is now apparent that the AGP converges exponentially to the Haar value as a function of the discrete time steps. The rate of convergence to the Haar value is given by
\begin{eqnarray}
\lambda = -\ln\left(1 - \frac{\mathcal{A}_{p}(U)}{\overline{\mathcal{A}}}\right),
\end{eqnarray}
which is strictly positive in this case. The numerical results shown in Fig.~\ref{fig:AGP-therm}(a) in the Supplemental Material illustrate clear agreement between the analytically obtained $\lambda$ and the numerical results.

\textbf{Case-2} ($\mathcal{A}_{p}(U)> \overline{\mathcal{A}}$):
When $\mathcal{A}_{p}(U) > \overline{\mathcal{A}}$, Eq.~(\ref{expasym_supp}) implies that the convergence to the Haar value is accompanied by alternating behavior. In this case, the evolution can be written as
\begin{eqnarray}
\left\langle \mathcal{A}_{p}(U^{(n)}) \right\rangle_{\tilde{F}}
&=& \overline{\mathcal{A}} \left[ 1 - (-1)^n \left( \frac{\mathcal{A}_{p}(U)}{\overline{\mathcal{A}}} - 1 \right)^n \right] \nonumber\\
&=& \overline{\mathcal{A}} \left[ 1 - (-1)^n \exp\left\{ n \ln\left( \frac{\mathcal{A}_{p}(U)}{\overline{\mathcal{A}}} - 1 \right) \right\} \right].
\end{eqnarray}
This behavior corresponds to an exponentially damped oscillatory relaxation towards the Haar-averaged value. Here, $\langle \mathcal{A}_{p}(U^{(n)})\rangle_{\tilde{F}}$ oscillates around the Haar value $\overline{\mathcal{A}}$ due to the presence of the sign factor $(-1)^n$.

}

\begin{figure}
\begin{center}
\includegraphics[scale=0.5]{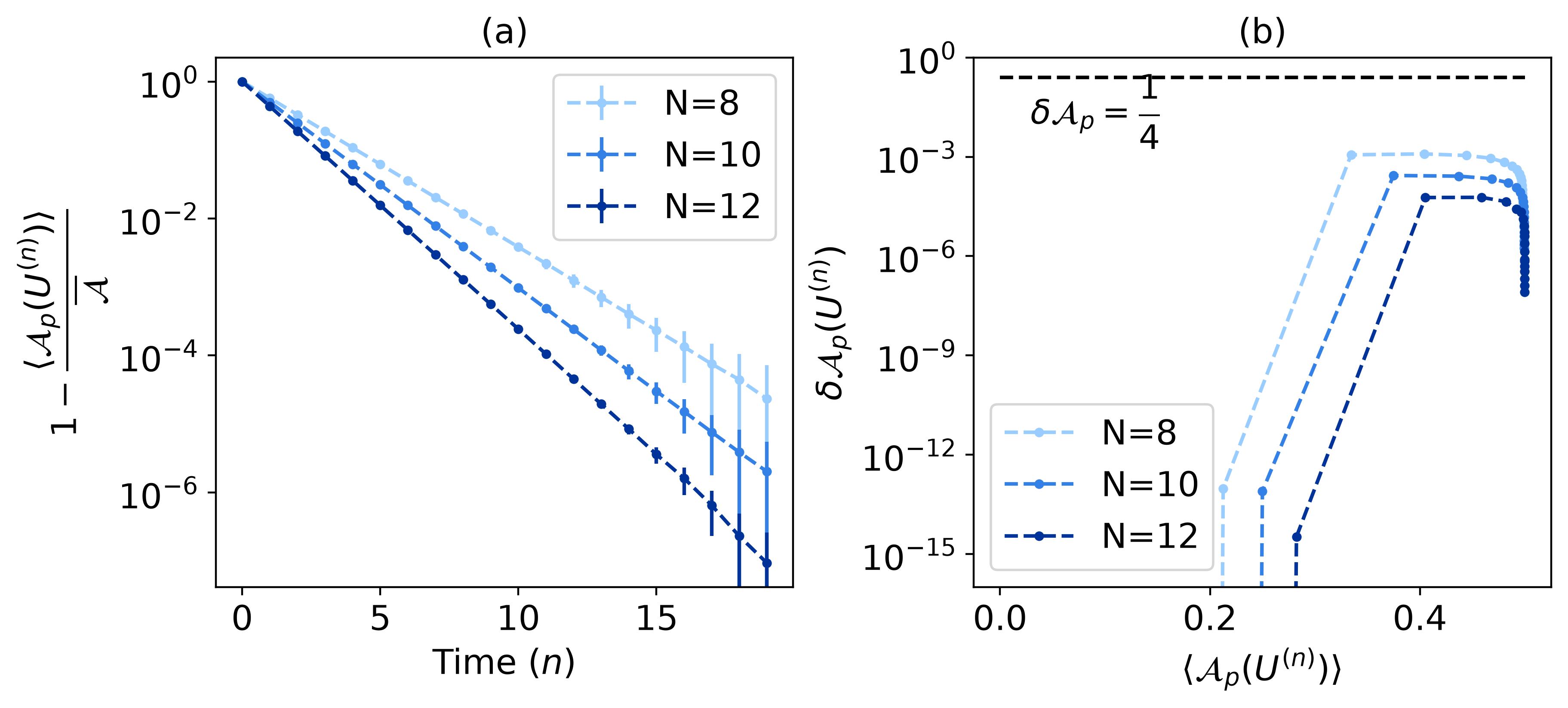}
\caption{\label{fig:AGP-therm}(a) Thermalization or equilibration of $Z_2$-AGP, characterized by the exponential decay of the quantity $1-\frac{\mathcal{A}_{p}(U^{(n)})}{\overline{\mathcal{A}}}$ for a fixed $U = u^{\otimes N}$ with $u = \exp\{-i\alpha\sigma_x\}$ and $\alpha = \pi/24$. Results are shown for three different system sizes, namely, $N=8$, $10$, and $12$. The numerical results (dots) coincide with the analytical expressions in Eq.~(6) of the main text and Eq.~(\ref{expasym_supp}) in the Supplemental Material (dashed lines). The error bars are given by the standard deviation of $\mathcal{A}_{p}(U^{(n)})$.
(b) Fluctuations of the $Z_2$-AGP, as measured by $\delta \mathcal{A}_{p}(U^{(n)})$ (see Eq.~(\ref{fluctAGP})), plotted versus the averaged $Z_2$-AGP on a semi-log scale for the same dynamics as in panel (a). The numerical results indicate that the fluctuations are increasingly suppressed with increasing system size. The black colored dashed line denotes the worst-case bound on the fluctuations, given by $\delta \mathcal{A}_{p}=1/4$. All numerical simulations are performed by sampling $10^3$ instances of the circuit dynamics for $N=8$ and $10$, and $10^2$ instances for $N=12$.}
\end{center}
\end{figure}

\subsubsection{{Details on fluctuations and typicality of Eq.~(\ref{expasym_supp})}}

While Eq.~(\ref{asym_impact}) concerns the averaged $Z_2$-AGP over sandwiched free unitaries, it is equally important to examine the fluctuations of $\mathcal{A}_{p}(VFU)$ as $F$ is sampled randomly from the group of free unitaries. We quantify these fluctuations using the variance $\delta^2 \mathcal{A}_{p}(VFU)$, which is defined as
\begin{eqnarray}\label{fluctAGP}
 \delta^2 \mathcal{A}_{p}(VFU) = \left\langle \mathcal{A}^{2}_{p}(VFU)  \right\rangle_{F} - \left\langle \mathcal{A}_{p}(VFU)  \right\rangle^2_{F}  .
\end{eqnarray}
The second term on the right-hand side can be readily evaluated using Eq.~(\ref{asym_impact}). However, evaluating the first term would require knowledge of the fourth moment of the free unitary group. Instead, we focus on obtaining a meaningful upper bound on this term.

{
First, we note that $\mathcal{A}_{p}$ for arbitrary unitaries is bounded as $0 \leq \mathcal{A}_{p}(VFU) \leq \tfrac{1}{2}$. Using standard extremal bounds for random variables with bounded support, it follows that the variance satisfies
\begin{equation}
\delta^2\mathcal{A}_{p}(VFU) \leq \langle \mathcal{A}_{p}(VFU) \rangle_{F} \left( \dfrac{1}{2} - \langle \mathcal{A}_{p}(VFU) \rangle_{F} \right).
\end{equation}
This bound depends only on the second moment and does not require knowledge of higher moments of the group of free unitaries. In the simplest case where $U=V$, using Eq.~(\ref{asym_impact}), we obtain
\begin{eqnarray}
\delta^2\mathcal{A}_{p}(UFU)
&\leq&
\mu(U)\left(\tfrac{1}{2} - \mu(U)\right),\;\; \mathrm{ where }\;  \mu(U) = 2\mathcal{A}_p(U) - \dfrac{\mathcal{A}^2_p(U)}{\overline{\mathcal{A}}}.
\end{eqnarray}
This bound provides useful insight into the behavior of fluctuations in $\mathcal{A}_{p}(VFU)$. In the limit $\mathcal{A}_{p}(U)=0$, the bound vanishes, implying that the fluctuations are suppressed. Similarly, when $\mathcal{A}_{p}(U)\rightarrow 1/2$, the bound again vanishes, indicating suppressed fluctuations. In contrast, fluctuations can be significant in the intermediate regime.
}

{We further note that the above bound implies a worst-case upper limit on the variance.
Since $0 \leq \mathcal{A}_p(VFU) \leq \tfrac{1}{2}$, it follows that $\delta \mathcal{A}_p(VFU) \leq {1}/{4}$. This result implies that the fluctuations do not grow indefinitely with increasing system size. We further confirm this by numerically computing the fluctuations of the $Z_2$-AGP for the dynamics considered in Fig.~\ref{fig:AGP-therm}(a). The corresponding results are shown in Fig.~\ref{fig:AGP-therm}(b). We find that the worst-case bound remains several orders of magnitude larger than the actual values. Moreover, the fluctuations clearly decay exponentially with increasing system size.
}

A similar line of reasoning can be applied to the other resource theories considered in this work.

\vspace{2em}

\section{Main results concerning the non-stabilizing power}
\label{app-Nonstab}

In this section, we give rigorous proofs of all results concerning the non-stabilizing power, and we also include a few known results for the sake of completeness.

\subsection{Derivation of the twirling identity in Eq.~(7) (main text)}
The linear non-stabilizing power of an arbitrary non-Clifford unitary $U$ is given by $m_p(U)=1-2^N\text{Tr}\left[ Q U^{\otimes 4}\overline{\left( |\psi\rangle\langle\psi | \right)^{\otimes 4}}U^{\dagger\otimes 4} \right]$, where $Q=\dfrac{1}{2^{2N}}\sum_{j=1}^{2^{2N}} P^{\otimes 4}_j$, with $P_j$s denoting Pauli strings supported over $N$ qubits. The operator $Q$ can be viewed as the projector onto the stabilizer subspace. We take $A=U^{\otimes 4}\rho_f U^{\dagger \otimes 4}$, where $\rho_f=\left\langle \left( |\psi\rangle\langle \psi | \right)^{\otimes 4} \right\rangle_{|\psi\rangle\in\mathrm{STAB}(2^N)}$. Then the fourth-order twirling channel (i.e., the uniform averaging over Clifford unitaries) is given by
\begin{eqnarray}
 \int_{C}d\mu(C)\left( C^{\dagger\otimes 4}\; U^{\otimes 4}\rho_f U^{\dagger \otimes 4}\;C^{\otimes 4} \right)   = \left(\alpha Q+ \alpha_{\perp}Q^{\perp}\right)\Pi^{(4)},
\end{eqnarray}
where $\Pi^{(4)}=\left(\sum_{j=1}^{24}\pi_j\right)/24$ and $Q^{\perp}=\mathbb{I}-Q$. The coefficients $\alpha$ and $\alpha_{\perp}$ are given by
\begin{eqnarray}
 \alpha = \dfrac{\text{Tr}\left( Q\Pi^{(4)} A \right)}{\text{Tr}\left( Q\Pi^{(4)} \right)}   = \dfrac{1-m_p(U)}{2^N\;\text{Tr}\left( Q\Pi^{(4)} \right)}
 \;\;\text{ and }\;\;
\alpha_{\perp}=\dfrac{\text{Tr}\left( Q^{\perp}\Pi^{(4)} A \right)}{\text{Tr}\left( Q^{\perp}\Pi^{(4)} \right)}  = \dfrac{m_p(U)-1+2^N}{2^N\;\text{Tr}\left( Q^{\perp}\Pi^{(4)} \right)} .
\end{eqnarray}
Upon substituting the above expressions, we get
\begin{eqnarray}\label{eqE3}
 \int_{C}d\mu(C)\left( C^{\dagger\otimes 4}\; A\;C^{\otimes 4} \right)   &=& \left(\dfrac{1-m_p(U)}{2^N\;\text{Tr}\left( Q\Pi^{(4)} \right)}\; Q+ \dfrac{m_p(U)-1+2^N}{2^N\;\text{Tr}\left( Q^{\perp}\Pi^{(4)} \right)}\;Q^{\perp}\right)\Pi^{(4)}\nonumber\\
 &=& \left\langle \left( |\psi\rangle\langle\psi | \right)^{\otimes 4}\right\rangle_{|\psi\rangle\in\mathrm{STAB}(2^N)} - m_p(U)\left( \dfrac{Q\Pi^{(4)}}{2^N\;\text{Tr}\left( Q\Pi^{(4)} \right)} -\dfrac{Q^{\perp}\Pi^{(4)}}{2^N\;\text{Tr}\left( Q^{\perp}\Pi^{(4)} \right)} \right),
\end{eqnarray}
where in the second equality, we used  $\left\langle \left( |\psi\rangle\langle\psi | \right)^{\otimes 4}\right\rangle_{|\psi\rangle\in\mathrm{STAB}(2^N)} = \dfrac{1}{2^N}\dfrac{Q\Pi^{(4)}}{\text{Tr}\left( Q\Pi^{(4)} \right)}+ \left( 1-\dfrac{1}{2^N} \right)\dfrac{Q^{\perp}\Pi^{(4)}}{\text{Tr}\left( Q^{\perp}\Pi^{(4)} \right)}$ \cite{leone2022stabilizer}. The second term on the right-hand side of the above equation can be simplified by taking $U$ to be a random unitary from the unitary group $\mathcal{U}(2^N)$ and performing the Haar average:
\begin{eqnarray}
\left( \dfrac{Q\Pi^{(4)}}{2^N\;\text{Tr}\left( Q\Pi^{(4)} \right)} -\dfrac{Q^{\perp}\Pi^{(4)}}{2^N\;\text{Tr}\left( Q^{\perp}\Pi^{(4)} \right)} \right)=\dfrac{1}{\overline{m}}\left[\left\langle \left( |\psi\rangle\langle\psi | \right)^{\otimes 4}\right\rangle_{|\psi\rangle\in\mathrm{STAB}(2^N)}- \dfrac{\Pi^{(4)}}{\mathrm{Tr}\left( \Pi^{(4)} \right)}\right].
\end{eqnarray}
Therefore, Eq.~(\ref{eqE3}) becomes
\begin{eqnarray}\label{E5}
\left\langle\left( C^{\dagger\otimes 4}\; U^{\otimes 4}\rho_f U^{\dagger \otimes 4}\;C^{\otimes 4} \right)\right\rangle_{C} = \left[ 1-\dfrac{m_p(U)}{\overline{m}} \right]\; \rho_f +\dfrac{m_p(U)}{\overline{m}}\; \dfrac{\Pi^{(4)}}{\mathrm{Tr}\left( \Pi^{(4)} \right)}   \,,
\end{eqnarray}
thereby yielding the twirling identity associated with the non-stabilizing power (see Eq.~(7) of the main text).

\subsection{Theorem~1 (main text): Non-stabilizing power version}
\label{app-Nonstab-1}

Similar to the case of the $Z_2$-AGP, one can extend Theorem~1 in the main text to the non-stabilizing power. Using the same protocol as before, the probability of obtaining a fixed measurement outcome in the QRT of non-stabilizerness is given by
\begin{eqnarray}
  p_b=\langle  \psi |U^{\dagger}C^{\dagger} |b\rangle\langle  b |CU |\psi\rangle,
\end{eqnarray}
where  $|\psi\rangle$ and $C$ denote a random stabilizer state and a random Clifford unitary, respectively. Since the linear non-stabilizing power is a fourth-order quantum state polynomial, we evaluate the fourth power of the probability, averaged over the free unitaries and the states,
\begin{eqnarray}
\langle p^{(4)}_b\rangle_{C, |\psi\rangle}&=&  \left\langle \text{Tr}\left[ \langle b^{\otimes 4}| \left( CU|\psi\rangle\langle\psi |U^{\dagger}C^{\dagger} \right)^{\otimes 4}|b^{\otimes 4}\rangle \right]\right\rangle_{C, |\psi\rangle} \nonumber\\
&=&\left[1-\dfrac{m_p(U)}{\overline{m}}\right] \; \text{Tr}\left[ \langle b^{\otimes 4}|\rho_f|b^{\otimes 4}\rangle  \right]+\dfrac{m_p(U)}{\overline{m}}\;\dfrac{\text{Tr}\left[ \langle b^{\otimes 4}| \Pi^{(4)} |b^{\otimes 4}\rangle \right]}{\mathrm{Tr}\left( \Pi^{(4)} \right)},
\end{eqnarray}
where, in the second equality, we used the result from Eq.~(\ref{E5}). For simplicity, we write $k_1=\dfrac{\text{Tr}\left[ \langle b^{\otimes 4}| \Pi^{(4)} |b^{\otimes 4}\rangle \right]}{\mathrm{Tr}\left( \Pi^{(4)} \right)}$ and $k_2=\text{Tr}\left[ \langle b^{\otimes 4}|\rho_f|b^{\otimes 4}\rangle  \right]$. Then, one can write the unbiased estimator of the non-stabilizing power as
\begin{eqnarray}
\langle m_{\mathrm{est}}  \rangle = \dfrac{k_2-\langle \tilde{p}^{(4)}\rangle_{C, |\psi\rangle}}{k_2-k_1},  \;\; \mathrm{ where }\;\; \tilde{p}^{(4)}= \dfrac{1}{2^{N_B}}\sum_{b=0}^{2^{N_B}-1}p^{4}_b.
\end{eqnarray}
For $N_A\geq 1$, the constants $k_1$ and $k_2$ take the forms
\begin{eqnarray}\label{k1nonstab}
k_1=\dfrac{\mathrm{Tr}\left( \langle b^{\otimes 4}|\Pi^{(4)}|b^{\otimes 4}\rangle \right)}{\mathrm{Tr}(\Pi^{(4)})}
= \prod_{s=0}^{3}\dfrac{2^{N_{A}}+s}{2^N+s} ,
\end{eqnarray}
and
\begin{eqnarray}\label{k2nonstab}
k_2=\mathrm{Tr}\left( \langle b^{\otimes 4}|\rho_f|b^{\otimes 4}\rangle \right)
= \dfrac{\mathrm{Tr}\left( Q_A\Pi^{(4)}_{A} \right)}{2^{N+N_B}\mathrm{Tr}\left( Q\Pi^{(4)} \right)}\; +\; \left( 1-\dfrac{1}{2^N} \right) \; \left[\dfrac{\mathrm{Tr}\left(\Pi^{(4)}_{A}\right) - \dfrac{\mathrm{Tr}\left( Q_A\Pi^{(4)}_{A} \right)}{2^{N_B}}}{\mathrm{Tr}\left( Q^{\perp}\Pi^{(4)} \right)}\right],
\end{eqnarray}
where $\mathrm{Tr}\left( Q\Pi^{(4)} \right)=\dfrac{(2^N+1)(2^N+2)}{6}$ and $\mathrm{Tr}\left( Q_A\Pi^{(4)}_A \right)=\dfrac{(2^{N_A}+1)(2^{N_A}+2)}{6}$ \cite{leone2022stabilizer}.

{
\subsubsection{Fluctuations of the estimator over the Clifford group and stabilizer states}
Here, we evaluate the fluctuations of the estimator and, in particular, provide analytical predictions for the special case when the given unitary is Haar random. Similar to the QRT of AGP, the variance can be computed using the expression
\begin{eqnarray}
\mathrm{Var}(m_{\mathrm{est}})
= \dfrac{1}{(k_2-k_1)^2}\left[ \left\langle \dfrac{1}{2^{2N_B}} \sum_{b, b' \in \{0, 1\}^{N_{B}}} p^4_{b}p^4_{b'} \right\rangle - \left\langle \tilde{p}^{(4)} \right\rangle^2  \right],
\end{eqnarray}
where $k_1$ and $k_2$ are the constants given in Eqs.~(\ref{k1nonstab}) and ~(\ref{k2nonstab}).

\begin{figure}
\begin{center}
\includegraphics[scale=0.5]{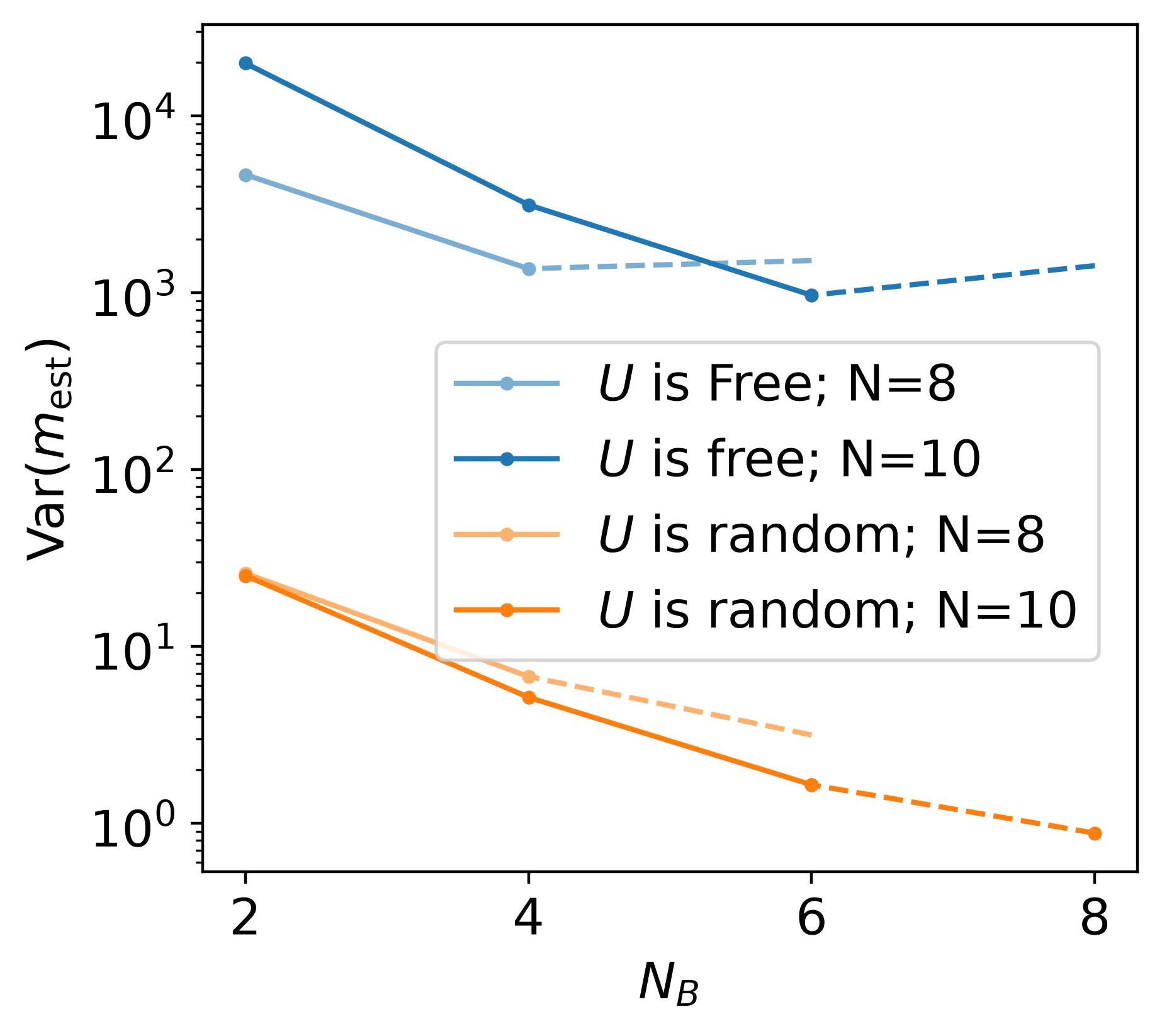}
\caption{\label{fig:mag-therm}  The figure illustrates numerical evaluation of the variance of the non-stabilizing power estimator as a function of the measured subsystem size $N_{B}$, for two limiting cases, namely, when $U$ is a random Clifford operator (blue) and when $U$ is a Haar-random unitary (orange). Results are shown for two system sizes, $N=8$ and $N=10$. The data are obtained by averaging over $\sim 10^4$ samples of random free (Clifford) operations and random stabilizer states. The results indicate that, for fixed $N$, the variance is exponentially suppressed with increasing $N_{B}$ in both cases. Moreover, the variance remains larger when the unitary remains close to being a Clifford.}
\end{center}
\end{figure}

\textbf{ $U$ is Haar random: }Here, we perform the Haar average over the unitary group to obtain the variance of the estimator. In this case, we note that $\langle \tilde{p}^{(4)}\rangle = k_1$. We identify that $\langle p^4_{b}p^4_{b'}\rangle_{U\in U{(2^N)}}$ for projection onto arbitrary states $|b\rangle$ and $|b'\rangle \in \mathcal{H}_{B}$ takes the following form:
\begin{eqnarray}
\left\langle p^4_{b}p^4_{b'} \right\rangle =
\begin{cases}
\displaystyle
\prod_{s=0}^{7}
\frac{2^{N_{A}} + s}{2^{N} + s},
& \text{if } \; b = b' \\[6pt]\\
\displaystyle k_1^2\,
\prod_{s=0}^{3}
\dfrac{2^N +s}{2^N+4+s},
& \text{otherwise},
\end{cases}
\end{eqnarray}
\begin{eqnarray}
\mathrm{Var}(m_{\mathrm{est}}) &=& \dfrac{1}{(k_2-k_1)^2} \left[ \dfrac{1}{2^{N_{B}}} \left( \prod_{s=0}^{7}
\frac{2^{N_{A}} + s}{2^{N} + s} \right) + \dfrac{2^{N_B}-1}{2^{N_{B}}}  \left(\prod_{s=0}^{3}
\dfrac{2^N +s}{2^N+4+s} \right) k^2_{1} - k^2_{1}  \right]\nonumber\\
&\approx & \dfrac{72}{2^{N_{B}}-1} \;\; \mathrm{ for }\;\; N\gg N_{B}.
\end{eqnarray}

}

\subsection{Non-stabilizing power counterpart of Eq. (2) (main text)}

Here, we examine the validity of Eq.~(2) of the main text for the QRT of non-stabilizerness. To this end, we compute the leading-order correction to the average non-stabilizerness of a projected state produced by the protocol under consideration, which is given by
\begin{equation}
\left\langle\text{Tr}\left[\left(\mathbb{I}_{2^{N_A}}- 2^{N_A}Q\right) \dfrac{  \left(\langle b |CU |\psi\rangle\langle\psi |U^{\dagger}C^{\dagger}|b\rangle \right)^{\otimes 4} }{\left(\langle\psi |U^{\dagger}C^{\dagger}|b\rangle\langle b| CU|\psi \rangle\right)^{4} }\right]\right\rangle_{C, |\psi\rangle}\approx \dfrac{\left\langle \text{Tr}\left[ \left(\mathbb{I}_{2^{N_A}}- 2^{N_A}Q\right) \left(\langle b |CU |\psi\rangle\langle\psi |U^{\dagger}C^{\dagger}|b\rangle \right)^{\otimes 4} \right]\right\rangle_{C, |\psi\rangle} }{\left\langle\left(\langle\psi |U^{\dagger}C^{\dagger}|b\rangle\langle b| CU|\psi \rangle\right)^{4}\right\rangle_{C, |\psi\rangle}}+\text{1-st order terms},
\end{equation}
where the average is performed over the Clifford operators ($C$) and the initial stabilizer states ($|\psi\rangle$). The numerator can be solved as follows:
\begin{eqnarray}
\left\langle \text{Tr}\left[ \left(\mathbb{I}_{2^{N_A}}- 2^{N_A}Q\right) \left(\langle b |CU |\psi\rangle\langle\psi |U^{\dagger}C^{\dagger}|b\rangle \right)^{\otimes 4} \right]\right\rangle_{C, |\psi\rangle}  =\dfrac{\text{Tr}\left(\Pi^{(4)}_{A}\right)}{\text{Tr}\left(\Pi^{(4)}_{AB}\right)} \dfrac{m_p(U)}{\overline{m_{AB}}}\overline{m_A},
\end{eqnarray}
The denominator, after the averaging, can be written as
\begin{eqnarray}
\left\langle\left(\langle\psi |U^{\dagger}C^{\dagger}|b\rangle\langle b| CU|\psi \rangle\right)^{4}\right\rangle_{C, |\psi\rangle}
&=&\left(1-\dfrac{m_p(U)}{\overline{m_{AB}}}\right)  \text{Tr}\left( \langle b^{\otimes 4}|\;\rho_f\;|b^{\otimes 4}\rangle  \right)+\dfrac{m_p(U)}{\overline{m_{AB}}}\dfrac{\text{Tr}\left(\Pi^{(4)}_{A}\right)}{\text{Tr}\left(\Pi^{(4)}_{AB}\right)} ,
\end{eqnarray}
where $\rho_f=\left\langle \left( |\psi\rangle\langle\psi | \right)^{\otimes 4} \right\rangle_{|\psi\rangle\in\mathrm{STAB}(2^N)}$. Recall from the previous subsection that $k_1=\left( \langle b^{\otimes 4}|\;\Pi^{(4)}\;|b^{\otimes 4}\rangle  \right)$ and $k_2=\mathrm{Tr}\left( \langle b^{\otimes 4}|\;\rho_f\;|b^{\otimes 4}\rangle  \right)$. Then, the leading-order correction to the average non-stabilizerness of the projected states becomes
\begin{eqnarray}
\left\langle m_\text{PS}\right\rangle_{C, |\psi\rangle} \approx \dfrac{\dfrac{m_p(U)}{\overline{m}}\;\overline{m_A}\; k_1}{\left( 1-\dfrac{m_p(U)}{\overline{m}} \right) k_2 + \dfrac{m_p(U)}{\overline{m}}   k_1}\approx \dfrac{m_p(U)}{\overline{m}}\;\overline{m_A}\;\; \mathrm{ when  }\;\; N_{B}\ll N.
\end{eqnarray}
Therefore, the resource content of the projected states behaves very similarly to the non-stabilizing power of the unitary $U$.

\subsection{Thermalization of the non-stabilizing power (Theorem~2 of the main text)}
As shown in Ref.~\cite{varikuti2025impact}, the non-stabilizing power of quantum circuits built from interlacing layers of Clifford and non-Clifford operations relaxes exponentially to the Haar-averaged value. For the sake of self-consistency, we present the derivation for the same in this section. Before proceeding, we first examine the non-stabilizing power of a random Clifford operation sandwiched between two arbitrary non-Clifford unitaries.
\begin{lemma}\label{lemma2}\normalfont{\cite{varikuti2025impact}}
Let $U$ and $V$ be two arbitrary {unitaries} in the QRT of non-stabilizerness over an $N$-qubit Hilbert space $\mathcal{H}^{2^{N}}$, with their corresponding finite non-stabilizing powers given by $m_{p}(U)$ and $m_p(V)$, respectively. Let $C$ be a random unitary drawn from the Clifford group according to its Haar measure. Then, the non-stabilizing power of $C$ sandwiched between $U$ and $V$, on average, is related to $m_{p}(U)$ and $m_{p}(V)$ as
\begin{equation}
\left\langle m_{p}(VCU) \right\rangle_{C} = m_{p}(U)+m_{p}(V)-\dfrac{m_{p}(U)m_{p}(V)}{\overline{m}}.
\end{equation}
\end{lemma}
\begin{proof}
We are interested in evaluating the quantity
\begin{eqnarray}
\left\langle m_p(VCU)  \right\rangle_{C} =1-2^N\int_{C}d\mu(C) \text{Tr}\left[ Q \left( VCU \right)^{\otimes 4}\overline{\left( |\psi\rangle\langle\psi \right)^{\otimes 4}} \left( VCU \right)^{\dagger\otimes 4}  \right]   .
\end{eqnarray}
From Eq.~(\ref{E5}), it follows that
\begin{eqnarray}
\int_{C}d\mu(C) (VCU)^{\otimes 4} \overline{\left( |\psi\rangle\langle\psi | \right)^{\otimes 4}} (VCU)^{\dagger \otimes 4} = \left[ 1-\dfrac{m_p(U)}{\overline{m}} \right] V^{\otimes 4}\overline{\left( |\psi\rangle\langle\psi | \right)^{\otimes 4}}V^{\dagger \otimes 4} + \dfrac{m_p(U)}{\overline{m}}\dfrac{\Pi^{(4)}}{\mathrm{Tr}\left( \Pi^{(4)} \right)}.
\end{eqnarray}
Then, the non-stabilizing power of the above configuration is given by
\begin{eqnarray}
\left\langle  m_p(VCU)   \right\rangle_{C}&=& \left[ 1-\dfrac{m_p(U)}{\overline{m}} \right] m_p(V) + m_p(U)\nonumber\\
&=&m_p(U)+m_p(V)-\dfrac{m_p(U)m_p(V)}{\overline{m}},
\end{eqnarray}
proving the lemma.
\end{proof}\\

We now state the result corresponding to the exponential convergence of $m_p$ to the Haar value under interlacing free and non-free dynamics.
\begin{theorem}(\normalfont{Theorem~2 for the non-stabilizing power})
Let $U^{(n)} = U C_{n-1} U \cdots C_1 U$, where each $C_j$ (for all $1 \leq j \leq n-1$) is an independently drawn random unitary from the Clifford group, sampled according to the Haar measure. Let $U$ be a fixed non-Clifford unitary with a finite non-stabilizing power given by $m_{p}(U)$. {Then, the evolution of non-stabilizing power satisfies}
\begin{equation}
\left\langle m_{p}(U^{(n)}) \right\rangle_{\tilde{C}} = \overline{m} \left[ 1 - \left[ 1 - \dfrac{m_{p}(U)}{\overline{m}} \right]^n \right],
\end{equation}
where the expectation $\langle \cdot \rangle_{\tilde{C}}$ denotes averaging over all the independent random Clifford operations $\{C_j\}$ at each time step.
\end{theorem}
\begin{proof}
Given an arbitrary non-Clifford operation $U$ with a finite non-stabilizing power $m_p(U)$, Lemma~\ref{lemma2} implies the following:
\begin{eqnarray}
\langle m_{p}(VCU) \rangle_{C} &=& \left( 1-\dfrac{m_p(U)}{\overline{m}} \right)m_p(V) + m_p(U)\nonumber\\
&=& \overline{m}\left[1- \left( 1-\dfrac{m_p(U)}{\overline{m}} \right) \left( 1-\dfrac{m_p(V)}{\overline{m}} \right) \right].
\end{eqnarray}
By recursively performing the integrations over the Clifford operations, we finally get
\begin{eqnarray}
\langle m_p(U^{(n)})\rangle_{\tilde{C}} = \overline{m}\left[ 1-\prod_{j=1}^{n}\left( 1-\dfrac{m_p(U_j)}{\overline{m}} \right) \right].
\end{eqnarray}
For identical non-Clifford unitaries, i.e., $U_j=U$ for all $j$, the above expression becomes
\begin{eqnarray}\label{expnon_stab_1}
\langle m_p(U^{(n)})\rangle_{\tilde{C}} = \overline{m}\left[ 1-\left( 1-\dfrac{m_p(U)}{\overline{m}} \right)^{n} \right].
\end{eqnarray}
This proves Theorem~2 for the non-stabilizing power.
\end{proof} \\

{
Here, we analyze the evolution of the non-stabilizing power in two regimes: $m_p(U)\leq \overline{m}$ and $m_p(U)>\overline{m}$. In the first case, Eq.~(\ref{expnon_stab_1}) can be expressed as
\begin{eqnarray}
\left\langle m_{p}(U^{(n)}) \right\rangle_{\tilde{C}}
&=& \overline{m} \left[ 1 - \exp\left\{ n \ln\left( 1 - \frac{m_{p}(U)}{\overline{m}} \right) \right\} \right].
\end{eqnarray}
Since $0 \leq \frac{m_{p}(U)}{\overline{m}} \leq 1$, the logarithm is negative, implying that the above expression approaches $\overline{m}$ exponentially fast as a function of the discrete time steps $t$. The associated convergence rate is
\begin{eqnarray*}
\lambda = -\ln\left(1 - \frac{m_{p}(U)}{\overline{m}}\right),
\end{eqnarray*}
which is strictly positive in this regime. On the other hand, when $m_p(U)> \overline{m}$, Eq.~(\ref{expnon_stab_1}) leads to a qualitatively different behavior. The evolution can be rewritten as
\begin{eqnarray}
\left\langle m_{p}(U^{(n)}) \right\rangle_{\tilde{C}}
&=& \overline{m} \left[ 1 - (-1)^n \exp\left\{ n \ln\left( \frac{m_{p}(U)}{\overline{m}} - 1 \right) \right\} \right].
\end{eqnarray}
In this case, the magnitude of the deviation from $\overline{m}$ still decays exponentially with $n$, while the factor $(-1)^n$ induces an alternating sign. Consequently, $\left\langle m_{p}(U^{(n)}) \right\rangle_{\tilde{C}}$ exhibits damped oscillations around the Haar-averaged value $\overline{m}$ before converging to it.
}

\newpage

{\textmd{\textit{In the remainder of this Supplemental Material, we generalize the main results to the resource theories of entanglement and coherence. These examples are included solely to illustrate the generality and broad applicability of our framework, rather than to introduce additional conceptual or technical overhead.}}}

\section{Extension to the QRT of entanglement}

To show the generality of our main results, we perform a similar analysis for the QRT of entanglement in this section and for the QRT of coherence in the next one.
In the resource theory of entanglement (across a bipartition $AB$), the set of local unitaries (and classical communications) constitutes the free operations. For pure bipartite states, the von Neumann entropy, which is given by $\mathcal{S}(|\psi\rangle)=-\text{Tr}\left[ \rho_A\ln\rho_A \right]$, is the gold-standard resource monotone. The corresponding linear entropy is given by the purity: $\mathcal{S}=1-\text{Tr}(\rho^2_A)$. Consider a bipartite quantum system with Hilbert space $\mathcal{H}_A \otimes \mathcal{H}_B$, where the subsystems $A$ and $B$ have dimensions $d_A$ and $d_B$, respectively. For a pure product state $|\psi\rangle = |\phi_A\rangle \otimes |\phi_B\rangle$, a bipartite unitary operator $U$ transforms it into the state $U|\psi\rangle$. The entanglement created by this operation can be quantified using the purity of the reduced state on subsystem $B$. Explicitly, if $\rho_B = \operatorname{Tr}_A\big(U|\psi\rangle\langle\psi| U^\dagger\big)$ is the reduced density matrix on $B$, then the entanglement measure based on purity is given by
\begin{eqnarray}
\mathcal{E}(U|\psi\rangle) &=& 1 - \operatorname{Tr}(\rho_B^2)\nonumber\\
& =& 1 - \operatorname{Tr}\left[ U^{\otimes 2} \left( |\phi_A \phi_B \rangle \langle \phi_A \phi_B| \right)^{\otimes 2} U^{\dagger \otimes 2} S_{BB'} \right],
\end{eqnarray}
where $S_{BB'}$ swaps the two copies of subsystem $B$. Building upon this, the \emph{entangling power} of the unitary operator $U$ is defined as the average amount of entanglement generated when $U$ acts on Haar-distributed random product states:
\begin{equation}
e_p(U) = \overline{\mathcal{E}(U|\psi\rangle)} = 1- \operatorname{Tr}\left( U^{\otimes 2} \overline{\left( |\phi_A\phi_B\rangle\langle\phi_A\phi_B | \right)^{\otimes 2}} U^{\dagger \otimes 2} S_{BB'} \right).
\end{equation}

Similar to the twirling identities we have obtained for the $Z_2$-AGP and non-stabilizerness, here we consider the twirling channel given by $\mathbb{T}=\left\langle F^{\otimes 2}\;U^{\otimes 2}\;\rho_f\; U^{\dagger\otimes 2}\; F^{\dagger\otimes 2} \right\rangle_{F}$, where $\rho_f=\left\langle\left(| \phi_A\phi_B\rangle\langle\phi_A\phi_B |\right)^{\otimes 2}\right\rangle_{|\phi_A\rangle, |\phi_B\rangle}$. We evaluate the identity in the following:
\begin{eqnarray}
\mathbb{T}&=&\left\langle F^{\otimes 2}\;U^{\otimes 2}\;\rho_f\; U^{\dagger\otimes 2}\; F^{\dagger\otimes 2} \right\rangle_{F}\nonumber\\
&=&\int_{u_A, u_B}d\mu(u_A)d\mu(u_B) \left( u_A\otimes u_B \right)^{\otimes 2} U^{\otimes 2} \overline{\left( |\phi_A\phi_B\rangle\langle\phi_A\phi_B | \right)^{\otimes 2}} U^{\dagger \otimes 2} \left( u^{\dagger}_A\otimes u^{\dagger}_B \right)^{\otimes 2}\nonumber\\
&=& \sum_{i, j\in\{+, -\}} \text{Tr}\left( \Pi^{i (2)}_{A} \Pi^{j (2)}_{B} U^{\otimes 2} \overline{\left( |\phi_A\phi_B\rangle\langle\phi_A\phi_B | \right)^{\otimes 2}} U^{\dagger \otimes 2}  \right)  \dfrac{\Pi^{i (2)}_{A} \Pi^{j (2)}_{B}}{\text{Tr}\left( \Pi^{i (2)}_{A} \Pi^{j (2)}_{B} \right)}\nonumber\\
&=& \left(1-\dfrac{e_p(U)}{2}\right) \overline{\left( |\phi_A\phi_B\rangle\langle\phi_A\phi_B | \right)^{\otimes 2}} + \dfrac{e_p(U)}{2} \dfrac{\Pi^{- (2)}_{A} \Pi^{- (2)}_{B}}{\text{Tr}\left( \Pi^{- (2)}_{A} \Pi^{- (2)}_{B} \right)}.
\end{eqnarray}
In the last equality, we used the relation $4\cdot\text{Tr}\left( \Pi^{i (2)}_{A} \Pi^{j (2)}_{B} U^{\otimes 2} \overline{\left( |\phi_A\phi_B\rangle\langle\phi_A\phi_B | \right)^{\otimes 2}} U^{\dagger \otimes 2} \right) = 1+ (-1)^{i+j} + (1-e_p(U))\left( (-1)^{i}+(-1)^{j} \right)$, where we assigned $0$ for the symbol $+$ and $1$ for $-$. With a slight adjustment, the above relation can be rewritten as
\begin{eqnarray}\label{ent-main}
\mathbb{T} = \left(1-\dfrac{e_p(U)}{\overline{e_p}}\right) \overline{\left( |\phi_A\phi_B\rangle\langle\phi_A\phi_B | \right)^{\otimes 2}} + \dfrac{e_p(U)}{\overline{e_p}} \dfrac{\Pi^{(2)}_{AB}}{\text{Tr}\left( \Pi^{ (2)}_{AB}  \right)},
\end{eqnarray}
where $\overline{e_p}$ denotes the Haar-averaged value for the entanglement generating power of the unitary $U$ and is given by $\overline{e_p}=1-(d_A+d_B)/(d_Ad_B+1)$.

\subsection{Theorem~1 (main text): Entangling power version}

We now consider the projected ensemble of a typical product state evolved under an arbitrary unitary followed by the random free operations. We consider the local measurements on the overlapping region of $A$ and $B$. Let the subsystem being measured be $A_2B_2$, where $A_2\subset A$ and $B_2\subset B$ (see Fig.~\ref{fig:ent_proj}(a) of this supplemental material). {We denote by $\{|a_2\rangle\}$ and $\{|b_2\rangle\}$ the computational basis states associated with the subsystems $A_2$ and $B_2$, respectively.}

\begin{figure}
\includegraphics[scale=0.55]{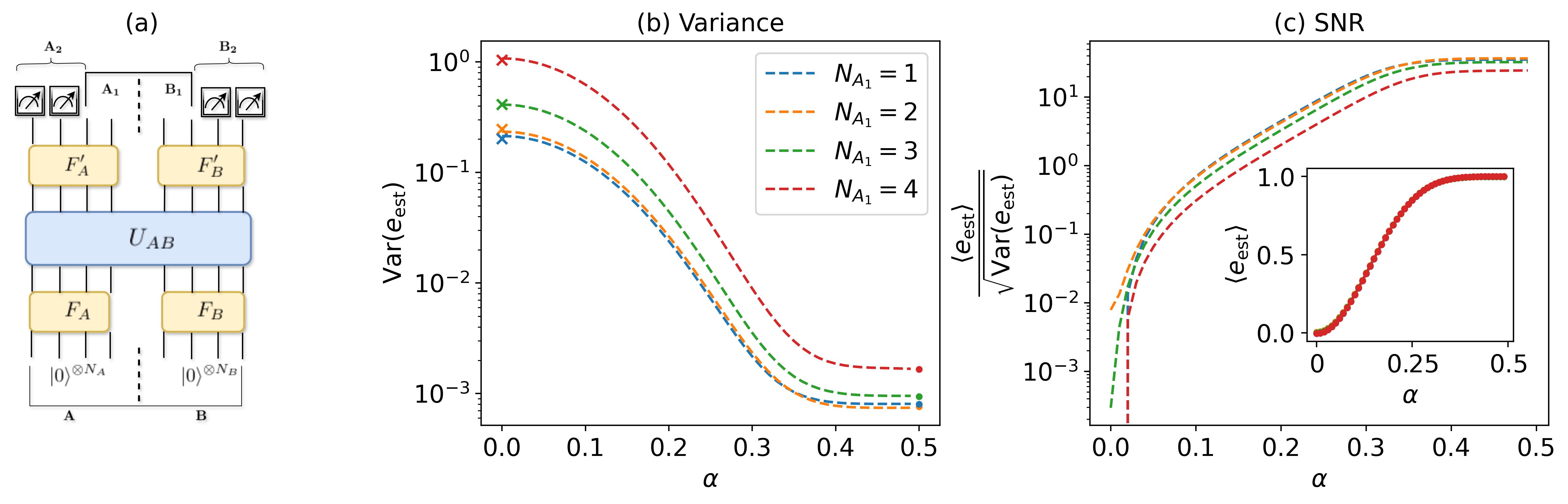}
\caption{\label{fig:ent_proj} (a) Schematic illustration of the protocol for quantifying the entangling power of a bipartite unitary $U_{AB}$, as outlined in the main text. The measurements at the end of the protocol are performed on the subsystems $A_2\subset A$ and $B_2\subset B$. For simplicity, in this work we consider $N_{A_2}=N_{B_2}$, which also implies $N_{A_1}=N_{B_1}$. (b) The variance of the estimator as a function of $\alpha$ when the entangling dynamics is given by $U_{AB} = e^{-i\alpha H}$, where $H$ is considered to be a random Hermitian operator ($H$) with $N=10$, and $N_{A}=N_{B}=5$. Here, we take a single random instance of $H$ for the numerical computations. The points marked by ``$\times$" denote the analytical predictions of the variance obtained in Eq.~(\ref{varentsep}). The points located at $\alpha=0.5$ marked by ``$\cdot$" indicate the variances corresponding to Haar random unitaries (see Eq.~(\ref{varenthaar})). (c) Signal-to-noise ratio (SNR) over random product unitaries and states, for the same random Hamiltonian considered in (b). The results in (b) and (c) are evaluated for over $\sim 5\times10^3$ samples of random free unitaries (local unitaries) and states (product states). }
\end{figure}

{
Following the circuit setting shown in Fig.~\ref{fig:ent_proj}(a), the probability of obtaining the measurement outcome corresponding to the computational basis state $(|a_2b_2\rangle)$ upon performing projective measurements on the subsystem $A_2B_2$ is
\[
p_{a_2b_2} = \mathrm{Tr}\left[\langle a_2b_2|(u_A\otimes u_B)U|\phi_A\phi_B\rangle\langle \phi_A\phi_B |U^{\dagger}(u_A\otimes u_B)^{\dagger} |a_2b_2\rangle\right].
\]
Squaring the above expression, performing the Haar averages over the free states and unitaries, and finally using Eq.~(\ref{ent-main}), one readily finds that the entangling power of $U$ across the bipartition $(AB)$ is related to the second moment of the outcome probability distribution as
\begin{equation}\label{eprotocol}
\langle e_{\mathrm{est}} \rangle = \dfrac{e_{p}(U)}{\overline{e_p}} = \dfrac{\langle p^2_{a_2b_2} \rangle -k_2 }{k_1 - k_2} = \dfrac{\langle \tilde{p}^{(2)} \rangle -k_2}{k_1-k_2},
\end{equation}
where
\begin{equation*}
e_{\mathrm{est}} = \dfrac{\tilde{p}^{(2)}-k_2}{k_1-k_2} \quad\text{ is the unbiased estimator, and }\quad
\tilde{p}^{(2)}=\dfrac{1}{2^{N_{A_2B_2}}}\sum_{\substack{
a_2 \in \{0,1\}^{N_{A_2}} \\
b_2 \in \{0,1\}^{N_{B_2}}
}} p^{2}_{a_2, b_2}.
\end{equation*}
We find the constants $k_1$ and $k_2$ to be
\begin{eqnarray}
k_1= \dfrac{\text{Tr}\left( \langle p^{\otimes 2}|\Pi^{(2)}_{AB}|p^{\otimes 2}\rangle \right)}{\text{Tr}\left( \Pi^{(2)}_{AB} \right)} = \dfrac{2^{N_{A_1B_1}}(2^{N_{A_1B_1}}+1)}{2^{N}(2^{N}+1)}
\end{eqnarray}
and
\begin{eqnarray}
k_2 =\dfrac{\text{Tr}\left(\langle p^{\otimes 2} | \Pi^{(2)}_{A}\Pi^{(2)}_{B} |p^{\otimes 2}\rangle \right)}{\text{Tr}\left( \Pi^{(2)}_{A}\Pi^{(2)}_{B} \right)} =  \dfrac{2^{N_{A_1}}(2^{N_{A_1}}+1)2^{N_{B_1}}(2^{N_{B_1}}+1)}{2^{N_{A}}(2^{N_A}+1) 2^{N_B}(2^{N_B}+1)}.
\end{eqnarray}
Note that the above relations for $k_1$ and $k_2$ hold only when the unmeasured subsystem size is non-vanishing. Equation~(\ref{eprotocol}) extends the applicability of Theorem~1 of the main text to the resource theory of entanglement.

\subsubsection{Fluctuations of the estimator over free states and free unitaries}
\label{entfluct}

Here, we analyze the fluctuations of the entangling power estimator in two limiting cases, namely, separable $(U_{AB}=u_{A}\otimes u_B)$ and Haar-random $U$. By examining these two limits, we identify a clear contrast in the behavior of the fluctuations, which directly affects the number of resources required to estimate the entangling power. We now discuss these two cases in detail.\\

\textbf{(a) $U$ is separable: }In the first case, when the unitary has vanishing entangling power, we characterize the fluctuations in the estimator using the variance. In this case, Eq.~(\ref{eprotocol}) implies $\langle p^2_{a_2, b_2}\rangle=\langle  \tilde{p}^{(2)}\rangle=k_2$ for any $a_2\in\{0, 1\}^{N_{A_2}}$ and $b_2\in\{0, 1\}^{N_{B_2}}$. Therefore, we write the variance as
\begin{eqnarray}\label{varentinit1}
\mathrm{Var}
&=& \dfrac{1}{(k_1-k_2)^2} \left[\dfrac{1}{2^{2N_{A_2B_2}}} \sum_{a_2, b_2, a'_2, b'_2}\langle p^2_{a_2, b_2}\cdot p^2_{a'_2b'_2} \rangle -k^{2}_{2} \right],
\end{eqnarray}
where $\langle \cdot\rangle$ denotes the averaging over the free states and free unitaries.
Since $U$ is a free unitary, one can write $p_{a_2b_2} = p_{a_2}\cdot  p_{b_2}$ and $p_{a'_2b'_2} = p_{a'_2}\cdot  p_{b'_2}$. One can further decouple the term $\langle p^2_{a_2, b_2}\cdot p^2_{a'_2b'_2} \rangle$ as
\begin{equation*}
 \langle p^2_{a_2, b_2}\cdot p^2_{a'_2b'_2} \rangle = \left\langle p^2_{a_2}p^2_{a'_2}\right\rangle_{u_A, |\phi_A \rangle} \cdot  \left\langle p^2_{b_2}p^2_{b'_2}\right\rangle_{u_B, |\phi_B\rangle}.
\end{equation*}
The terms on the right-hand side can be computed through Haar averages over the subsystems $A$ and $B$, and are given by
\begin{eqnarray}
\langle p^2_{a_2}p^2_{a'_2}\rangle =
\begin{cases}
\displaystyle \prod_{s=0}^{3}
\frac{2^{N_{A_1}} + s}{2^{N_{A}} + s},
& \text{if } a_2 = a'_2 \\[6pt]\\
\displaystyle
\dfrac{2^{N_{A_1}}(2^{N_{A_1}}+1)2^{N_{A_1}}(2^{N_{A_1}}+1)}{2^{N_{A}}(2^{N_{A}}+1)(2^{N_{A}}+2)(2^{N_{A}}+3)}
,
& \text{otherwise},
\end{cases}
\end{eqnarray}
and similarly for $\langle p^2_{b_2}p^2_{b'_2}\rangle$. It then follows from Eq.~(\ref{varentinit1}) that
\begin{eqnarray}\label{varentsep}
\mathrm{Var}
&=& \dfrac{1}{(k_1-k_2)^2} \left[\dfrac{1}{2^{4N_{A_2}}}\left\{2^{N_{A_2}} \langle p^4_{a_2}\rangle+2^{N_{A2}}\left(2^{N_{A_2}}-1\right) \langle p^2_{a_2}\cdot p^2_{a'_2}\rangle\right\}^2    -k^{2}_{2} \right]
\end{eqnarray}
Here, we make an assumption that $N_{A}=N_{B}$ and $N_{A_2}=N_{B_2}$. Then, in the limit of large $N$, the variance follows the scaling
\begin{equation}
\mathrm{Var}\approx \dfrac{1}{2^{N_{A_{2}}}-1}.
\end{equation}

\textbf{(b) $U$ is Haar random: }On the other hand, when the unitary $U$ is a Haar unitary, the second moment of the probability becomes $\langle p^2_{a_2, b_2}\rangle=\langle  \tilde{p}^{(2)}\rangle=k_1$ for any $a_2\in\{0, 1\}^{N_{A_2}}$ and $b_2\in\{0, 1\}^{N_{B_2}}$. Then, the variance of the estimator becomes
\begin{eqnarray}\label{varentinit}
\mathrm{Var} =  \dfrac{1}{(k_1-k_2)^2} \left[\dfrac{1}{2^{2N_{A_2B_2}}} \sum_{a_2, b_2, a'_2, b'_2}\langle p^2_{a_2, b_2}\cdot p^2_{a'_2b'_2} \rangle -k^{2}_{1} \right].
\end{eqnarray}
Note that in the above expression $\langle \cdot\rangle$ also includes the averaging over $U\in U(2^N)$. Upon performing the Haar average over $U\in U(2^N)$, we obtain the terms inside the summation as follows:
\begin{align}
\left\langle p^2_{a_2, b_2}\, p^2_{a'_2, b'_2} \right\rangle
&=
\begin{cases}
\displaystyle \prod_{s=0}^{3}
\frac{2^{N_{A_1 B_1}} + s}{2^{N} + s},
& \text{if } a_2 = a'_2,\; b_2 = b'_2 \\[6pt]\\
\displaystyle k_1^2\,
\frac{2^N (2^N+1)}{(2^N+2)(2^N+3)},
& \text{otherwise}.
\end{cases}
\end{align}
By incorporating the above equation in Eq.~(\ref{varentinit}), we get the variance of the estimator as
\begin{eqnarray}\label{varenthaar}
 \mathrm{Var} = \dfrac{1}{(k_2-k_1)^2}\left[\dfrac{1}{2^{N_{A_2B_2}}}  \left( \prod_{s=0}^{3} \frac{2^{N_{A_1B_1}} + s}{2^{N} + s}  \right) +\left\{ \left(\dfrac{2^{N_{A_2B_2}}-1}{2^{N_{A_2B_2}}}\right)\left( \dfrac{2^{N}(2^N+1)}{(2^N+2)(2^{N}+3)}\right) -1\right\} k^2_1 \right].
\end{eqnarray}
Similar to the previous case, we take the simplest scenario where $N_{A}=N_{B}$ and $N_{A_1}=N_{B_1}$. Then, further simplifications in the large $N$ limit reveal that
\begin{equation}\label{varenthaar2}
\mathrm{Var} \approx \dfrac{1}{2\cdot 2^{N}} \;  + \; \dfrac{1}{2^N (2^{N_{A_2}}-1)},
\end{equation}
indicating that the fluctuations get suppressed exponentially with the system size $N$.

}

{
\subsubsection{Towards experimental implementation---A concrete illustration}
In the following, we provide a concrete illustration of our protocol for quantifying entangling power, thereby illustrating its experimental feasibility. Entanglement is one of the most extensively studied resources in quantum many-body physics. In this setting, entangling power is defined as the average entanglement generated by a unitary evolution acting on typical pure product states, quantified by the linear entropy.
The main questions regarding the feasibility of our protocol concern $(1)$ the required experimental operations and $(2)$ the necessary number of free unitaries to obtain converged results. 

\emph{(1) Required experimental operations:}
As illustrated in Fig.~\ref{fig:ent_proj}(a), the protocol requires
\begin{itemize}

\item the ability to apply a non-free unitary---this is the object of interest, whose entangling power should be estimated, and which by assumption can be applied;
\item the ability to apply random free operations---this is the primary bottleneck; in the present case, these correspond to random unitaries acting on subsystems, which thanks to recent advances in quantum technologies have already been demonstrated in large-scale experiments, even beyond straightforward classical tractability \cite{boixo2018characterizing, doi:10.1126/science.aau4963, PhysRevLett.120.050406};
\item the ability to projectively measure in the computational basis---this is a standard operation in any universal quantum computer \cite{DiVincenzo_2000} and is also available in ultracold-atom setups thanks to advanced quantum-gas microscopes \cite{bakr2009quantum, sherson2010single, gross2021quantum}.

\end{itemize}
It remains to show that the required number of free unitaries remains tractable. The variance analysis from the previous subsection demonstrates that the estimator is remarkably {favorable}.\\

\emph{(2) Required number of free unitaries:}
As shown in the previous subsection (see Sec.~\ref{entfluct} of this Supplemental Material), in the large-system limit under the assumptions $N_A=N_B$ and $N_{A_1}=N_{B_1}$, and in the regime of vanishing entangling power, the variance of the estimator is suppressed exponentially with the size of the measured subsystem. Consequently, Chebyshev's inequality implies that the number of free unitaries required to estimate the entangling power with additive accuracy $\epsilon$ scales as
$M
\sim
\frac{1}{\epsilon^2(2^{N_{A_2}}-1)}.
$
Thus, in this regime, the {number of free unitaries required} is determined by the measured subsystem size rather than the total system size.
Note that the estimator's variance is suppressed even further as the dynamics become chaotic. In particular, Eq.~(\ref{varenthaar2}) shows that, in the large-$N$ limit, the variance decays exponentially with the total system size $N$, implying an even lower {number of random free unitaries} than in the vanishing-entangling-power regime.

Taken together, these observations suggest that our protocol is highly viable on current experimental architectures.
Besides the fundamental interest in characterizing quantum effects, the protocol also provides a route to certify quantum hardware that generates complex many-body quantum dynamics.  \\
}

\subsection{Entangling power counterpart of Eq.~(2) (main text)}

One can further show that the average purity of a projected state, up to the leading term, can be evaluated as
\begin{eqnarray}
\mathbb{E}_{|\psi_A\rangle, |\psi_B\rangle, u_A, u_B} \mathcal{E}(|\phi_{A_1B_1}\rangle) \approx \dfrac{ e_p(U) }{\overline{e_{AB}}}\overline{e_{A}} \dfrac{k_1}{\left( 1-\dfrac{e_p(U)}{\overline{e_{AB}}} \right)k_2+ \dfrac{e_p(U)}{\overline{e_{AB}}}k_1}
\end{eqnarray}

For simplicity, we take $N_A=N_B$ and $N_{A_1}=N_{B_1}$. Then, one can immediately see that
\begin{eqnarray}
 \dfrac{k_2}{k_1} &=&  \dfrac{1}{1+\dfrac{2\cdot  2^{N_A}}{2^N+1}}\; +\; \dfrac{2}{\left( 1+\dfrac{2\cdot 2^{N_A}}{2^N+1} \right)\left( \dfrac{1+2^{2N_{A_1}}}{2^{N_{A_1}}} \right)}\nonumber\\
 &\approx & 1\; +\; \dfrac{2}{\left( \dfrac{1+2^{2N_{A_1}}}{2^{N_{A_1}}} \right)} \;\;\hspace{2em}\mathrm{ for }\;\; N\gg 1.
\end{eqnarray}
It then follows that for even moderately large $N_{A_1}$, we get $k_2/k_1\approx 1 + 2/2^{N_{A_1}}+O(2^{-2\cdot N_{A_1}})$. Therefore, the average purity of the projected states is given by
\begin{eqnarray}
\mathbb{E}_{|\psi_A\rangle, |\psi_B\rangle, u_A, u_B} \mathcal{E}(|\phi_{A_1B_1}\rangle) \approx \dfrac{ e_p(U) }{\overline{e_{AB}}}\overline{e_{A}} \,,
\end{eqnarray}
indicating that the entanglement generated in the projected states closely tracks the entangling power of the unitary $U$ taken in the protocol. This analysis demonstrates that our main results are generic and that the QRT of entanglement fully reproduces the behavior outlined in the main text.

\subsection{Thermalization of entanglement generating power(Theorem ~2 of the main text)}

To make the supplemental material self-contained, we also briefly discuss the thermalization of the entanglement-generating power under alternating free and non-free unitaries, in direct analogy with the other QRTs considered in this work. For further details, the interested reader may refer to Refs.~\cite{jonnadula2017impact, jonnadula2020entanglement}. As in the other QRTs, the entangling power is affected by the free operations, which in this setting correspond to local unitaries.

Given two arbitrary unitaries $U$ and $V$ with finite entangling power, we consider the quantity $e_p(VFU)$, where the free operation is $F = u_A \otimes u_B$. Averaging over all free operations yields
\begin{equation}
\left\langle e_p (V F U) \right\rangle_F
= e_p(U) + e_p(V) - \frac{e_p(U)e_p(V)}{\overline{e_p}} .
\end{equation}
This result can be naturally generalized to generic quantum circuits composed of alternating free and non-free operations. Defining $ U^{(n)} = U F_{n-1} \cdots U F_2 \, U F_1$,
and averaging independently over the free operations at each time step, we obtain
\begin{equation}
\left\langle e_p(U^{(n)}) \right\rangle_{F_1, \ldots, F_{n-1}}
= \overline{e_p} \left[ 1 - \left(1 - \frac{e_p(U)}{\overline{e_p}} \right)^n \right].
\end{equation}

This expression is strikingly similar to the thermalization law presented in Theorem~2 of the main text, further strengthening the generality of our results.

\section{QRT of Coherence}

In this section, we extend our main results to the resource theory of quantum coherence. Here, the free states are the computational basis vectors, and the free unitaries are taken to be of the form $\pi\; V_{\mathrm{diag}}$, where $\pi$ indicates an arbitrary basis permutation operator and $V_{\mathrm{diag}}$ is a random diagonal unitary. Given a state $|\psi\rangle\in \mathcal{H}^{d}$, the coherence with respect to the computational basis can be computed using the quantity
\begin{equation}
 C(|\psi\rangle)=1- \sum_{i=0}^{d-1}\left| \langle i|\psi\rangle \right|^{4}.
\end{equation}
This quantity vanishes if and only if $|\psi\rangle$ is a computational basis vector and hence faithfully quantifies the resource. It is easy to see that the above quantity is a second-degree quantum state polynomial. Having defined the coherence in a state, the coherence generating power (CGP) of an arbitrary unitary $U$ with respect to the computational basis can be defined as \cite{anand2021quantum, varikuti2022out, styliaris2019quantum}
\begin{eqnarray}
\mathcal{C}(U)=1-\dfrac{1}{\text{dim}(U)} \sum_{i, j=1}^{\text{dim}(U)} \left |\langle i|U|j\rangle\right |^{4}  .
\end{eqnarray}
Before proceeding, we note that
\begin{equation}
V_\mathrm{diag}=\sum_{k=0}^{\text{dim}(U)-1} e^{i\phi_k}|k\rangle\langle k|.
\end{equation}
We first consider the twirling channel as in the previous cases:
\begin{eqnarray}
(FU)^{\otimes 2}\rho (FU)^{\dagger \otimes 2} &=& \left( \pi\; V_{\mathrm{diag}}\; U \right)^{\otimes 2}\rho_f \left( \pi\; V_{\mathrm{diag}}\; U \right)^{\dagger \otimes 2} \nonumber\\
&=&\sum_{k, l, k', l' =0}^{\text{dim}-1}\exp\left\{ i\left( \phi_k+\phi_l-\phi_{k'}-\phi_{l'} \right) \right\}\; \left(\pi^{\otimes 2} |kl\rangle\langle k'l'|\pi^{\dagger\otimes 2}\right)\; \left( \langle kl | U^{\otimes 2}\rho_f U^{\dagger\otimes 2}  |k'l'\rangle \right).
\end{eqnarray}
To compute the average over the free operations, we first perform the averaging over the diagonal unitaries,
\begin{eqnarray}\label{H3}
\left\langle (FU)^{\otimes 2}\rho (FU)^{\dagger \otimes 2} \right\rangle_{F} &=& \left\langle\sum_{k, l}  \pi^{\otimes 2}|kl\rangle\langle kl|\pi^{\dagger\otimes 2} \left( \langle kl | U^{\otimes 2}\rho_f U^{\dagger\otimes 2}  |kl\rangle \right)  + \sum_{k, l}  \pi^{\otimes 2}|kl\rangle\langle lk|\pi^{\dagger\otimes 2} \left( \langle kl | U^{\otimes 2}\rho_f U^{\dagger\otimes 2}  |lk\rangle \right)\right. \nonumber\\
&&\left.- \sum_{k}  \pi^{\otimes 2}|kk\rangle\langle kk|\pi^{\dagger\otimes 2} \left( \langle kk | U^{\otimes 2}\rho_f U^{\dagger\otimes 2}  |kk\rangle \right)\right\rangle_{\pi\in \mathrm{Perm}(\mathrm{dim}(U))}.
\end{eqnarray}
We now evaluate the averaging over the basis permutation operators for all the terms on the right-hand side separately. For fixed indices $k,l$, the permutation averages over $\pi \in \mathrm{Perm}(d)$ give
\begin{equation}
\left\langle \pi^{\otimes 2} |kl\rangle \langle kl| \pi^{\dagger \otimes 2} \right\rangle_{\pi\in\mathrm{Perm}(d)}=\dfrac{1}{d!} \sum_{\pi \in \mathrm{Perm}(d)} \pi^{\otimes 2} |kl\rangle \langle kl| \pi^{\dagger \otimes 2}
=
\begin{cases}
\dfrac{1}{d(d-1)} \displaystyle\sum_{i \neq j} |ij\rangle \langle ij| & (k \neq l), \\[2ex]
\dfrac{1}{d} \displaystyle\sum_{i} |ii\rangle \langle ii| & (k = l).
\end{cases}
\end{equation}
Similarly, for the swapped term we obtain
\begin{equation}
\left\langle \pi^{\otimes 2} |kl\rangle \langle lk| \pi^{\dagger \otimes 2} \right\rangle_{\pi\in\mathrm{Perm}(d)}=\dfrac{1}{d!} \sum_{\pi \in \mathrm{Perm}(d)} \pi^{\otimes 2} |kl\rangle \langle lk| \pi^{\dagger \otimes 2}
=
\begin{cases}
\dfrac{1}{d(d-1)} \displaystyle\sum_{i \neq j} |ij\rangle \langle ji| & (k \neq l), \\[2ex]
\dfrac{1}{d} \displaystyle\sum_{i} |ii\rangle \langle ii| & (k = l).
\end{cases}
\end{equation}
The third term becomes
\begin{eqnarray}
\left\langle \pi^{\otimes 2} |kk\rangle \langle kk| \pi^{\dagger \otimes 2} \right\rangle_{\pi\in\mathrm{Perm}(d)}     = \dfrac{1}{d}\sum_{i}|ii\rangle\langle ii|.
\end{eqnarray}
Consequently, Eq.~(\ref{H3}) can be rewritten as
\begin{eqnarray}
\left\langle (FU)^{\otimes 2}\rho (FU)^{\dagger \otimes 2} \right\rangle_{F} = \dfrac{1-T}{d(d-1)}  \left( \mathbb{I}+\mathrm{SWAP} \right)\; +  \; d\left( \dfrac{T}{d}-\dfrac{2-2T}{d(d-1)} \right)\; \rho_f\;,
\end{eqnarray}
where $T=\sum_{k}\langle kk| U^{\otimes 2}\rho_f U^{\dagger\otimes 2} |kk\rangle$. By noticing that $\Pi^{(2)}=\dfrac{\mathbb{I}+\mathrm{SWAP}}{2}$, and the Haar-averaged CGP given by $\overline{C}=\dfrac{d-1}{d+1}$, the twirling channel in the above equation can be written in a compact form as follows:
\begin{eqnarray}\label{H10}
\left\langle (FU)^{\otimes 2}\rho (FU)^{\dagger \otimes 2} \right\rangle_{F} =\left[1- \dfrac{C_p(U)}{\overline{C}}\right]\; \rho_f \; +\; \dfrac{C_p(U)}{\overline{C}}\; \dfrac{\Pi^{(2)}}{\mathrm{Tr}\left( \Pi^{(2)} \right)}\;.
\end{eqnarray}
The above equation displays remarkable similarity with the twirling identities obtained for the other QRTs in this work, such as $Z_2$-asymmetry (see Eq.(5) of the main text), non-stabilizerness (see Eq.~(7) of the main text), and the entangling power (see Eq.~(\ref{ent-main})). Hence, all the results discussed for the above cases can also be extended to the QRT of coherence.

\subsection{Thermalization of Coherence Generating Power (Theorem~2 of the main text)}
For the sake of completeness, we analyse thermalization of the CGP under interlacing free and non-free operations. From Eq.~(\ref{H10}), it is straightforward to verify the following expression depicting the impact of free operations on the CGP:
\begin{eqnarray}
\langle C_p(VFU) \rangle_{F} = C_p(U)+C_p(V)-\dfrac{C_p(U)C_p(V)}{\overline{C}},
\end{eqnarray}
where $U$ and $V$ are non-free operations, and the averaging is performed over the free operations $F$. When the non-free operations are identical, the above equation becomes
\begin{eqnarray}
\langle C_p(UFU)\rangle_{F}   & =& 2 C_p(U) -\dfrac{C^2_p(U)}{\overline{C}}\nonumber\\
&=& C_p(U)\left[ 2-\dfrac{C_p(U)}{\overline{C}} \right].
\end{eqnarray}
If we denote $U^{(n)}=UF_{n-1}\cdots UF_2UF_1 U$, independent averages over the free operations at every time step yield
\begin{eqnarray}
\left\langle  C_p(U^{n})\right\rangle_{F_1, \cdots F_{n-1}} = \overline{C}\left[ 1-\left( 1-\dfrac{C_p(U)}{\overline{C}} \right)^n \right]\;,
\end{eqnarray}
implying the expected exponential thermalization of the coherence-generating power toward its Haar-averaged value. This result further reinforces the findings presented in the main text.

\end{document}